\documentclass{article}

\usepackage{PRIMEarxiv}

\usepackage[utf8]{inputenc} 
\usepackage[T1]{fontenc}    
\usepackage{hyperref}       
\usepackage{url}            
\usepackage{booktabs}       
\usepackage{amsfonts,amsmath,amssymb,amsthm}       
\usepackage{mathtools}
\usepackage{nicefrac}       
\usepackage{microtype}      
\usepackage{lipsum}
\usepackage{fancyhdr}       
\usepackage{graphicx}       
\usepackage{natbib}
\graphicspath{{media/}}     

\usepackage{enumitem}
\usepackage{tikz} 
\usetikzlibrary{arrows.meta}
\DeclareMathOperator*{\argmax}{arg\,max}

\newtheorem{theorem}{Theorem}
\newtheorem{define}{Definition}
\newtheorem{lemma}{Lemma}
\newtheorem{assume}{Assumption}
\newtheorem{coro}{Corollary}

\title{Sequential Persuasion Using Limited Experiments
}

\author{
  Bonan Ni \\
  Tsinghua University
   \And
  Weiran Shen \\
  Renmin University of China
   \And
  Pingzhong Tang \\
  Tsinghua University and TuringSense
}

\begin{document}
\maketitle

\begin{abstract}
Bayesian persuasion and its derived information design problem has been one of the main research agendas in the economics and computation literature over the past decade. However,
when attempting to apply its model and theory, one is often limited by the fact that the sender can only implement very restricted information structures. Moreover, in this case, the sender can possibly achieve higher expected utility by performing a sequence of feasible experiments,
where the choice of each experiment depends on the outcomes of all previous experiments. Indeed, it has been well observed that real life persuasions often take place in rounds during which the sender exhibits experiments/arguments sequentially.

We study the sender's expected utility maximization using finite and infinite sequences of experiments. For infinite sequences of experiments, we characterize the supremum of the sender's expected utility using a function that generalizes the concave closure definition in the standard Bayesian persuasion problem. With this characterization, we first study a special case where the sender can use feasible experiments to achieve the optimal expected utility of the standard Bayesian persuasion without feasibility constraints, which is a trivial utility upper bound, and establish structural findings about the sender's optimal sequential design in this case. Then we derive conditions under which the sender's optimal sequential design exists; when an optimal sequential design exists, there exists an optimal design that is Markovian, i.e., the choice of the next experiment only depends on the receiver's current belief.
\end{abstract}

\section{Introduction}

Information design studies how the choice of one or more senders' information disclosure affects one or more receivers' actions, which in turn affects the senders' realized utilities. Bayesian persuasion is a canonical model of information design \cite{bayesian-persuasion}, which studies the interaction between one sender and one receiver. The sender announces and publicly performs an experiment whose outcome depends on an unknown world state. After knowing the experiment outcome, the receiver updates its belief, which is a distribution of the world state, using the Bayes rule, Then the receiver chooses an action to maximize their own expected utility under the posterior distribution of the world state. As shown by \citet{bayesian-persuasion}, every experiment can be characterized by a distribution over the receiver's posterior beliefs which is Bayes plausible, that is, the expectation of the distribution is equal to the prior distribution. The sender can choose among all Bayes plausible distributions over posteriors. A mapping from the posterior to the sender's utility summarizes the receiver's optimal action choice and the random realization of the world state. under an assumption about the receiver's tie-breaking rule, the sender's optimal experiment choice exists, and the optimal expected utility is given by the concave closure of the mapping from the posterior to the sender's utility. 

In practice, however, it is possible that the sender's optimal experiment is infeasible to implement, as recently mentioned by \citet{open-issue}. Various factors may restrict the sender's feasible experiments. First of all, the sender's technology and resource may be limited, so they cannot implement experiments that arbitrarily disclose information. Secondly, the receiver's technology to interpret the experiment may also be limited, or the receiver may incur a cost of interpretation, so that it is only incentivized to correctly update its belief for a particular, easy to interpret, subset of experiments. Various regulations and restrictions by third parties may also limit the sender's implementable experiments. For example, privacy laws may restrict the informativeness of a data provider's service. All these possibilities motivates us to study a more general problem of the sender's expected utility maximization using a limited set of feasible experiments. In this paper, we limit the set of feasible experiments by considering a subset of all distributions over the posteriors, so that the sender can only take experiments in the given subset. As we discuss in Section \ref{sec-restrict-discuss}, limiting feasible experiments to a particular set of distributions over posteriors is more general than limiting the feasible signaling schemes, which can be described by a set of probabilities of sending different signals in each world state. To see this, a set of distributions over posteriors represents not only feasible signaling schemes but also the set of applicable priors for each feasible signaling scheme. For example, whether a data provider is allowed to disclose the zip code of a certain user can depend on the population of corresponding users: For a very smaller number of users living in very different areas, the zip code can be used to exactly identify a user.

With a restricted subset of feasible experiments, the sender can possibly obtain higher expected utilities by taking more than one experiment. 
We consider the case where the sender takes multiple experiments sequentially, where the choice of every experiment depends on all the experiments taken previously and their outcomes. For a sequence of experiments taken by the sender, the receiver sequentially updates its belief according to the outcome of each experiment using the Bayes rule, and chooses an action after the sequence of experiments finishes. An experiment in a sequence is Bayes plausible if its expectation is equal to the belief induced by the experiment preceding it, and the sender can take an extra experiment after all previous experiments as long as the new experiment is both feasible and Bayes plausible. 

The theme of the paper is how the sender can maximize their expected utility using sequential choices of experiments. We study the sender's utility maximization problem from the following three perspectives:
\begin{itemize}
    \item How much expected utility can the sender obtain?
    \item Does the sender's optimal sequential choice rule exist?
    \item When the optimal choice rule exists, what structural findings can be drawn from it?
\end{itemize}
We answer the first question by studying the supremum of the sender's expected utility over all sequential choices of experiments, which is not necessarily attainable by some particular choice rule. When the sender can take arbitrarily many experiments, we provide a characterization of the supremum utility in Theorem \ref{thm-limit-characterize}, which is in the form of a generalization of the concave envelope for the Bayesian persuasion model. Then we address the second and the third question for a special case, where some optimal choice rule achieves an expected utility that is equal to the optimal one of Bayesian persuasion without feasibility constraints, so that the sender can use multiple feasible experiments to ``implement'' an optimal information structure which is not necessarily in the feasible set. We utilize an assumption that the set of feasible experiments is closed, and discover connections between the implementation of the optimal information structure and the sender's ability to lead the receiver's belief to some special destinations in the belief space with probability arbitrarily close to 1. Such connections are summarized in Theorem \ref{thm-implement}. The last part of this paper addresses the second and the third questions by deriving conditions for the existence of the sender's optimal sequential choice rule, and show that when some optimal choice rule exists, there exists an optimal choice rule which is Markovian, that is, the choices only depend on the receiver's current belief. Under the assumption of closed feasible sets, we derive a sufficient and necessary condition for the existence of the sender's optimal choice rule, which is stated as Theorem \ref{thm-opt-exist}.

\section{Related works}

Information design is an important topic of the economics and computation literature. Methodological findings establish linkage between information structure and the solution concepts of game theory \cite{unified-perspective, BCE}, as well as other general theoretical structures and properties \cite{concave-game,population-belief}. The design of information structure brings new possibility to classic models, including auctions \cite{fist-price,auction-inter,Myerson} and congestion games \cite{congestion-1,congestion-2}. Some other works study the design and pricing of the sender's signal \cite{selling-1,selling-2}.

It is mentioned by \citet{open-issue} that, the sender can have a constrained set of feasible information structures under various circumstances in practice, including bank stress tests \cite{bank-1,bank-2,bank-3}, quality certification \cite{quality-test-1,quality-test-2}, and clinical trials \cite{clinic-trail-1}. A more abstract feasibility constraint arises from a communication channel with limited capacity, which constraints the informativeness of the sender's feasible information structure \cite{noisy-persuasion,limited-entropy}. In addition, the sender can only use certain constrained information structures if it has chance to misreport the experiment outcome, therefore the equilibrium information structures are restricted by the sender's incentive compatibility \cite{cheap-talk-1,cheap-talk-2,cheap-talk-3} or credibility constraints \cite{credible-persuasion}. In contrast to these models with relatively specific feasibility constraints, some other works study the sender's persuasion problem with the set of feasible experiments restricted in more general ways. \citet{experimental-persuasion} study the sender's optimal persuasion using a given subset of all message schemes, and an alternative setting where the sender can garble the chosen experiment's result in a committed way. \citet{post-ante} study the sender's optimal and approximately optimal information structure under ex post and ex ante feasibility constraints. 

A number of works study sequential persuasion under various settings. \citet{sequential-trial} combine multi-phase experiments with restrictions on the sender's feasible information structure. In contrast to their work, we do not assume binary world state and allow the sender to take an arbitrary number of experiments, and we restrict feasible information structure using an given set of feasible experiments rather than some exogenously determined experiments in the sequence. A widely studied case where sequential persuasion can be necessary is the information design with multiple senders, for which the equilibrium can be given by multiple experiments sequentially chosen by different senders \cite{multiple-sender-1,multiple-sender-2,multiple-sender-3,multiple-sender-4,multiple-sender-5}. \citet{bilateral-trade} study multi-round communication between agents which improves the efficiency of a bilateral trade. \citet{private-experiment} study the sender's sequential private experiments with strategic disclosure to the receiver. Sequential persuasion also arises when the receiver interacts with a dynamic environment for multiple rounds \cite{dynamic-1}.

\section{Preliminary}

A finite set $\Omega$ contains all possible world states. Denote by $\Delta(\Omega)$ the set of probability distributions over $\Omega$. A world state $\omega \in \Omega$ is drawn from distribution $\mu \in \Delta(\Omega)$. Consider total variation distance on $\Delta(\Omega)$ to obtain a compact metric space, and denote by $\Delta(\Delta(\Omega))$ the set of Borel probability measures on $\Delta(\Omega)$. Throughout this paper we focus on those finite-support elements of $\Delta(\Delta(\Omega))$, which can be represented as follows: for integer $m \ge 1$, denote by $[m]$ the set $\{1, 2, \ldots, m\}$ (in particular $[0] = \emptyset$), then the set of probability distributions over $\Delta(\Omega)$ with finite support is given by:
\[
    F_0 \coloneqq \left\{(\lambda_j, p_j)_{j \in [m]} \in ((0, 1] \times \Delta(\Omega))^{m} ; m \ge 1, \sum_{j \in [m]} \lambda_j = 1, p_j \ne p_{j'} \text{ for every }1 \le j < j' \le m\right\}.
\]
Note that for convenience and without loss of generality, we let $p_j \ne p_{j'}$, but do not impose any restrictions on the ordering of index $j$. 

A sender and a receiver share $\mu$ as their prior. The sender can take some experiments whose outcome depends on the world state, and the receiver updates their belief according to the experiment outcome using the Bayes rule. By \citet{bayesian-persuasion}, an experiment is given by a mean-preserving spread of the Dirac distribution $\delta_p$ that assigns probability 1 to the receiver's original belief $p \in \Delta(p)$. For convenience, we use the term ``a mean-preserving spread of $p$'' to refer to a mean-preserving spread of the Dirac distribution $\delta_p$, and say an experiment spreads $p$ if it is a mean-preserving spread of $\delta_p$. For every $e = (\lambda_j, p_j)_{j \in [m]} \in F_0$, denote by $\tau(e) \coloneqq \{p_j; j \in [m]\}$ the support of $e$, and by $\sigma(e) \coloneqq \sum_{j \in [m]} \lambda_jp_j$ the expectation of $e$. Given the receiver's belief $p_0 \in \Delta(\Omega)$ before the experiment, the set of finite-support mean-preserving spreads of $p_0$ is given by
\[
    \{e \in F_0; \sigma(e) = p_0\},
\]
and after the experiment $e = (\lambda_j, p_j)_{j \in [m]}$, the receiver's new belief is given by: for every $j \in [m]$, the new belief becomes $p_j$ with probability $\lambda_j$. The outcome of experiment $e$ is the index $j$ whose corresponding $p_j$ becomes the receiver's new belief.

\subsection{Feasible experiments}
Throughout this paper, we restrict the sender's feasible experiments to some $F \subseteq F_0$. Therefore, when the receiver's belief is $p_0 \in \Delta(\Omega)$, the set of experiments the sender can take is given by
\[
    \mathcal{F}(p_0) \coloneqq \{e \in F; \sigma(e) = p_0 \}.
\]
When $e = (1, p)$ for some $p \in \Delta(\Omega)$, the experiment is called a trivial experiment because it does not reveal any additional information. Define $T \coloneqq \{(1, p) | p \in \Delta(\Omega)\}$ to be the set of trivial experiments. Note that a trivial experiment is equivalent to not taking any experiment which is always feasible for the sender. Therefore, without loss of generality, we assume $T \subseteq F$.

\subsection{Sequential persuasions}
To compensate for the loss of restricting the feasible information structure to $F$, we allow the sender to choose a sequence of experiments in $F$, where the choice of each experiment depends on the outcome of all previous experiments. 
\begin{define}
    For any $n \in \mathbb{N}$, an $n$-step history is a finite sequence $(p_0, e_1, p_1, \ldots, e_n, p_n)$, where $(p_i)_{0 \le i \le n}$ and $(e_i)_{i \in [n]}$ are separately sequences in $\Delta(\Omega)$ and $F$, and $\sigma(e_i) = p_{i-1}$ holds for every $i \in [n]$.
\end{define}
For every $\xi = (p_0, e_1, p_1, \ldots, e_n, p_n)$, we define operator $\text{last}(\xi) \coloneqq p_n$. For two sequences $s_1$ and $s_2$, we use $s_1 \oplus s_2$ to denote their concatenation. 
\begin{define}
    For any $n \in \mathbb{N}$, an $n$-step sequential persuasion starting from $\mu$ is a finite sequence $(\Xi_0, \phi_1, \Xi_1, \ldots, \phi_n, \Xi_{n})$ where:
    \begin{itemize}
        \item $\Xi_0 = \{(\mu)\}$ is the singleton set that contains the zero-step history with the starting point $\mu$. For every $i \in [n]$, $\Xi_i$ is a set of $i$-step histories;
        \item For every $i \in [n]$, $\phi_i: \Xi_{i-1} \mapsto F$ assigns to each $(i-1)$-step history the next experiment such that, for any $\xi \in \Xi_{i-1}$, $\phi_i(\xi) \in \mathcal{F}(\text{last}(\xi))$ holds;
        \item For every $i \in [n]$, $\Xi_{i} = \{ \xi \oplus (\phi_i(\xi), p); \xi \in \Xi_{i-1}, p \in \tau(\phi_i(\xi))\}$.
    \end{itemize}
    An infinite sequential persuasion starting from $\mu$ is an infinite sequence $(\Xi_0, \phi_1, \Xi_1, \ldots)$ that satisfies the above conditions for all $i \ge 1$.
\end{define}
For any given $(F, \mu)$, denote by $\mathcal{S}^{(n)}_\mu$ the set of $n$-step sequential persuasions starting from $\mu$, and by $\mathcal{S}_\mu$ the set of infinite sequential persuasions starting from $\mu$.

Although we allow the sender to perform infinite-length sequential persuasions, the sender cannot achieve its utility (to be defined later) unless the receiver chooses its action, which only happens if the sender finishes all experiments in finite steps. Therefore, we must define the termination of infinite sequential persuasions. The possibility of finite-step termination does not contradict to the infinite length of sequential persuasion: for example, the sequential persuasion can terminate in every step $i$ with probability $\delta \in (0,1)$, so that it almost surely terminates in finite steps, while the maximum number of possible steps is infinity.
\begin{define}\label{define-terminate}
    For any $n \in \mathbb{N}$ and $0 \le i \le n$, an $n$-step sequential persuasion $(\Xi_0, \phi_1, \Xi_1, \ldots, \phi_n, \Xi_n)$ terminates after history $\xi = (p_0, e_1, p_1, \ldots, e_i, p_i) \in \Xi_i$ if it constantly chooses the trivial persuasion afterward, that is, $\phi_{i+l}\big (\xi \oplus ((1, p_i), p_i)^{l-1} \big) = (1, p_i)$ for every $l \in [n-i]$. An infinite sequential persuasion $(\Xi_0, \phi_1, \Xi_1, \ldots)$ terminates after history $\xi = (p_0, e_1, p_1, \ldots, e_i, p_i) \in \Xi_i$ if $\phi_{i+l}\big (\xi \oplus ((1, p_i), p_i)^{l-1} \big) = (1, p_i)$ for every $l \ge 1$.
\end{define}
Here $((1, p_i), p_i)^{l-1}$ denotes the sequence of length $(2l-2)$ that repeats $((1, p_i), p_i)$ for $(l-1)$ times. Note that $[0] = \emptyset$, and by definition, every $n$-step sequential persuasion terminates after every $n$-step history in its $\Xi_n$.

An infinite sequential persuasion $(\Xi_0, \phi_1, \Xi_1, \ldots)$, or an $n'$-step sequential persuasion 
for some $n' \ge n$, introduces a random realization of $n$-step history by the randomness of experiment outcomes. The corresponding probability distribution over $\Xi_n$ is given by: for every $n$-step history $(p_0, e_1, p_1, \ldots, e_{n}, p_{n}) \in \Xi_n$, 
\begin{equation} \label{eq-history-distrib}
    \Pr[(p_0, e_1, p_1, \ldots, e_{n}, p_{n})] = \prod_{j=1}^{n} e_j(p_j),
\end{equation}
where with some slight abuse of notation, for every $e_i = (\lambda^{(i)}_j, p^{(i)}_j)_{j \in [m_i]}$, we define $e_i(p_i)$ to be $\lambda^{(i)}_j$ if $p_i = p^{(i)}_j$ for some $j$, and zero otherwise. Moreover, denote by $b_{n}$ the receiver's random belief after $n$ steps, the probability distribution of $b_n$ is induced by the probability distribution over $\Xi_n$:
\begin{equation} \label{eq-post-distrib}
    \Pr[b_n = p] = \sum_{\xi \in \Xi_{n}}\mathbb{I}[\text{last}(\xi) = p] \cdot \Pr[\xi].
\end{equation}
Note that all experiments taken have finite support, therefore every $\Xi_n$ is finite and the summation is always well-defined.

The distribution over $\Xi_n$ naturally induces a probability of termination after $n$ steps, and a probability of terminating at every particular belief $p \in \Delta(\Omega)$ after $n$ steps:
\begin{align*}
    & \Pr[S (S^{(n')}) \text{ terminates after } n \text{ steps}] = \sum_{\xi \in \Xi_n} \mathbb{I}[S (S^{(n')}) \text{ terminates after } \xi] \cdot \Pr[\xi], \\
    & \Pr[S (S^{(n')}) \text{ terminates after } n \text{ steps at belief } p] \\
    =& \sum_{\xi \in \Xi_n} \mathbb{I}[S (S^{(n')}) \text{ terminates after } \xi] \cdot \mathbb{I}[\text{last}(\xi) = p] \cdot \Pr[\xi].
\end{align*}
It is easy to show that the probability of termination after $n$ steps is increasing in $n$, therefore for any infinite sequential persuasion, the limiting probability of termination exists and is equal to the supremum:
\[
    \lim_{n \to \infty} \Pr[S \text{ terminates after } n \text{ steps}] = \sup_{n \in \mathbb{N}} \Pr[S \text{ terminates after } n \text{ steps}] 
\]

\subsection{Actions and utilities}
The receiver's action space is a compact set $A$. A bounded, continuous function $U: \Omega \times A \rightarrow \mathbb{R}$ gives the receiver's utility for every combination of action and world state. After knowing the outcome of all experiments to obtain some belief $p \in \Delta(\Omega)$, the receiver chooses action $\alpha(p) \in A$ to maximize its expected utility under $p$:
\[
    \alpha(p) \in \argmax_{a \in A} \sum_{\omega \in \Omega} p(\omega)U(\omega, a),
\]
where the existence of optimal action is guaranteed by the compactness of $A$. The receiver's action choice and utility realization can be summarized by the following mapping:
\[
    u(p) \coloneqq \sum_{\omega \in \Omega} p(\omega)U(\omega, \alpha(p)), \forall p \in \Delta(\Omega).
\]

The sender's utility is given by a bounded, continuous function $V: \Omega \times A \rightarrow [\underline V, \overline V]$ that maps the world state and the receiver's action choice to the sender's realized utility. 
We assume that $\underline V > 0$: Since the receiver only chooses action after the sender finishes all the experiments, the termination of the sender's sequential persuasion outweighs whether the exact value of the sender's realized utility is large enough. As shown in Lemma \ref{lem-AS-terminate}, even though the sender is allowed to take infinitely many experiments, they are incentivized to do so only with negligible probability. 
Also, we can assume that $\alpha$ breaks ties by choosing the candidate that the sender prefers most. Therefore, the receiver's action choice and the sender's utility realization can be summarized by the mapping 
\[
    v(p) \coloneqq \sum_{\omega \in \Omega} p(\omega)V(\omega, \alpha(p)), \forall p \in \Delta(\Omega),
\]
which also takes value in $[\underline V, \overline V]$ and is upper-semicontinuous, as shown by \citet{bayesian-persuasion}.

Throughout this paper, we focus on the sender's expected utility maximization for a given instance $(\mu, F, v)$. Our main results only assume bounded and upper-semicontinuous $v$, but the examples we provide contain $v$ that can be obtained from some trivial (and even finite) $A$ and the corresponding $(U,V)$.

Define function $\mathcal{V}$ that maps every sequential persuasion to its expected utility:
\begin{define}\label{define-func-v}
    The expected utility of an $n$-step sequential persuasion $S^{(n)} = (\Xi_0, \phi_1, \Xi_1, \ldots, \phi_n, \Xi_{n})$ is
    \[
        \mathcal{V} (S^{(n)}) \coloneqq \mathbb{E} [v(b_n)],
    \]
    where $b_n$'s distribution is given by Equation \eqref{eq-post-distrib}. The expected utility of an infinite sequential persuasion $S = (\Xi_0, \phi_1, \Xi_1, \ldots)$ is
    \[
        \mathcal{V} (S) \coloneqq \sup_{n \in \mathbb{N}} \sum_{\xi \in \Xi_n} \left\{\mathbb{I}[S \text{ terminates after } \xi] \cdot \Pr[\xi] \cdot v(\text{last}(\xi))\right\}.
    \]
\end{define}
With our assumption $\underline V > 0$, the following lemma tells that every infinite sequential persuasion which does not terminate with probability 1 is strictly suboptimal:
\begin{lemma}\label{lem-AS-terminate}
    For any $S \in \mathcal{S}_\mu$, if 
    \[
        \lim_{n \to \infty} \Pr[S \text{ terminates after } n \text{ steps}] < 1,
    \]
    then there exists $n \in \mathbb{N}$ and $S'^{(n)} \in \mathcal{S}^{(n)}_\mu$ such that $\mathcal{V}(S'^{(n)}) > \mathcal{V}(S)$.
\end{lemma}

\subsection{Discussion}
\subsubsection{Restriction over experiments}\label{sec-restrict-discuss}
In the standard Bayesian persuasion model, an experiment is constructed by a signaling scheme $\pi: \Omega \rightarrow \Delta(\mathcal{D})$ for some dictionary of message values $\mathcal{D}$, and the receiver updates its belief based on the realized message using the Bayes rule. For the receiver's initial belief $p$, any message scheme with a finite $\mathcal{D}$ generates a new belief according to a mean-preserving spread of $p$, which is an element of $F_0$. 

Instead of setting restrictions on feasible $\pi$'s, we restrict the sender's feasible experiments by subset $F \subseteq F_0$ because such restrictions are more general. If we only restrict the set of feasible $\pi$ with a finite dictionary, by default each $\pi$ can be used to spread every $p \in \Delta(\Omega)$ and obtain a set of elements in $F_0$. Therefore, every set of feasible $\pi$ corresponds to a subset of $F_0$, but not every $F \subseteq F_0$ can be given by a set of feasible $\pi$. In general, a subset $F \subseteq F_0$ can be given by the combination of (i) a set of feasible message schemes and (ii) for each message scheme $\pi$, a subset $D_\pi \subseteq \Delta(\Omega)$ restricting the set of beliefs that the sender can use $\pi$ to spread. 

\subsubsection{Dependence on history}\label{sec-discuss-markov}
In our definition of sequential persuasion, we allow the sender's experiment choice $\phi_j, j \ge 1$ to depend on the full history of past experiments and outcomes. We show in Theorem \ref{thm-opt-exist} that, when some optimal infinite sequential persuasion exists, there exists an optimal infinite sequential persuasion whose experiment choices only depend on the receiver's current belief, i.e. $\phi_j(\xi) = \phi_{j'}(\xi')$ as long as $\text{last}(\xi) = \text{last}(\xi')$, which we call a ``Markov'' sequential persuasion. 

For any fixed $n > 1$, the optimal $n$-step sequential persuasion in general has to choose experiments according to not only the current belief, but also the number of steps $j$. To see this, consider the following example: 

Let $\Omega = \{0,1\}$, so that we can represent every belief $p \in \Delta(\Omega)$ by $t_p = p(0) \in [0,1]$. Let $v(t) = \mathbb{I}[t=0] + \mathbb{I}[t=1], \forall t \in [0,1]$. $\mu = \frac13$. In addition to all trivial experiments, $F$ contains four elements: $e_1$ spreads $\frac13$ into $0$ and $\frac23$ both w.p. $\frac12$, $e_2$ spreads $\frac13$ into $0$, $\frac12$ and $1$ separately w.p. $\frac12$, $\frac13$ and $\frac16$, and the two symmetric experiments: $e_3$ spreads $\frac23$ into $\frac13$ and $1$ both w.p. $\frac12$, $e_4$ spreads $\frac23$ into $0$, $\frac12$ and $1$ separately w.p. $\frac16$, $\frac13$ and $\frac12$. For any $n$, the optimal $n$-step sequential persuasion can be characterized by: For $j < n$, take $e_1$ if $t_{p_j} = \frac13$, take $e_3$ if $t_{p_j} = \frac23$. For $j = n$, take $e_2$ if $t_{p_j} = \frac13$, take $e_4$ if $t_{p_j} = \frac23$.

\subsubsection{Utility after termination}\label{sec-terminate-utility}
We require an infinite sequential persuasion to realize utility only after termination. An alternative approach to define infinite-round utility is  
\[
    \mathcal{V'} (S) \coloneqq \sup_{n \in \mathbb{N}} \mathbb{E} [v(b_n)].
\]
In fact, the supremum of infinite-round utility $v_\infty$, to be defined in Section \ref{sec-limit-function}, does not change if we replace $\mathcal{V}$ by $\mathcal{V'}$ in its definition, as the functions $(v_n)_{n \in \mathbb{N}}$ and $v_\infty$ satisfy Lemma \ref{lem-iterate}, \ref{lem-converge} and Theorem \ref{thm-limit-characterize} for both definitions of $v_\infty$. The distinction between the two definitions affects the existence of an optimal infinite sequential persuasion. This can be illustrated by the following example, in which an infinite sequential persuasion maximizes $\mathcal{V'}$, but the supremum of $\mathcal{V}$ is not attainable by any infinite sequential persuasion: 

Let $\Omega = \{0,1\}$, and each belief $p \in \Delta(\Omega)$ is represented by $p(0)$. $v(t) = 2\left|t-\frac12\right|, \forall t \in [0,1]$. $\mu = \frac12$. $F$ is given by:
\[
    F \coloneqq T \cup \{(\lambda_j, p_j)_{j = 0, 1} \in F_0; 0 < p_0 < p_1 < 1, H_2(\lambda_0 p_0 + \lambda_1 p_1) = 2(\lambda_0 H_2(p_0) + \lambda_1 H_2(p_1))\},
\]
where $H_2(t) \coloneqq -t \log t - (1-t) \log (1-t), t \in (0,1)$ is the entropy of a binary variable. For any infinite sequential persuasion that keeps choosing nontrivial experiments, the expected entropy of $\omega$ decreases by half for every step. By taking enough steps of nontrivial experiments and then terminating, the expected utility given by both $\mathcal{V}$ and $\mathcal{V'}$ can be arbitrarily close to 1, therefore an infinite sequential persuasion is optimal if and only if its expected utility is equal to 1.

Any $S$ with $\mathcal{V}(S) > 0$ has to terminate with probability $\delta>0$ at some belief $0 < t < 1$. Therefore 
\[
    \mathcal{V}(S) \le 1 - \delta \cdot \left(1- 2\left|t - \frac12\right|\right) < 1,
\]
and such $S$ is sub-optimal. However, consider an infinite sequential persuasion $S'$ that keeps choosing nontrivial experiments and never terminates: one can easily see that $\mathcal{V'}(S') = 1$ and $S'$ maximizes $\mathcal{V'}$.

\section{Utility supremum of sequential persuasions}\label{sec-limit-function}
When $F = F_0$, any finite-step sequential persuasion can be merged into a single finite-support experiment in $F_0$. For an upper-semicontinuous $v$, the optimal expected utility over all information structures is given by the concavification of $v$, which can be attained by an experiment with support size at most $|\Omega|$. Therefore, the sender only needs to use a single experiment in $F$. When $F \subsetneq F_0$, however, it is generally necessary to use sequential persuasions to achieve a higher expected utility. This immediately raises the following question: how much utility can the sender achieve using sequential persuasions? we answer this question by showing that the supremum of the sender's expected utility is the minimum of a set of real-valued functions above $v$, which is in the form of a generalization of the concave envelope function.

\begin{define}
    For every $n \in \mathbb{N}$, $v_n: \Delta(\Omega) \rightarrow \mathbb{R}$ is the mapping from prior to the supremum of the sender's expected utility by taking an $n$-step sequential persuasion:
    \[
        v_n(p) \coloneqq \sup_{S^{(n)} \in \mathcal{S}^{(n)}_p} \mathcal{V}(S^{(n)}), \forall p \in \Delta(\Omega).
    \]
    In particular $v_0 = v$. 
    
    $v_\infty: \Delta(\Omega) \rightarrow \mathbb{R}$ is the mapping from prior to the supremum of the sender's expected utility by taking an infinite sequential persuasion:
    \[
        v_\infty(p) \coloneqq \sup_{S \in \mathcal{S}_p} \mathcal{V}(S), \forall p \in \Delta(\Omega).
    \]
\end{define}
The following lemma shows that the sequence $(v_n)_{n \in \mathbb{N}}$ can be obtained recursively. Note that we assume $F$ contains all trivial experiments, therefore $(v_n)_{n \in \mathbb{N}}$ is pointwise non-decreasing. In another word, we allow the sender to ``waste'' one step by taking a trivial experiment if they find an extra step to be useless for utility maximization.
\begin{lemma}\label{lem-iterate}
    The sequence $(v_n)_{n \in \mathbb{N}}$ satisfies 
    \begin{equation} \label{eq-iterate}
        v_{n}(p) = \sup_{(\lambda_j, p_j)_{j \in [m]} \in \mathcal{F}(p)} \sum_{j=1}^m \lambda_j v_{n-1}(p_j), \ \forall n \ge 1, \forall p \in \Delta(\Omega).
    \end{equation}
    In particular, $(v_n(p))_{0 \ge 1}$ is non-decreasing for every $p \in \Delta(\Omega)$.
\end{lemma}
Moreover, since $v$ is upper bounded by $\overline V$, and for every finite sequential persuasion, the value of $\mathcal{V}$ is a weighted sum of $v$, every $v_n$ is also upper bounded by $\overline V$. By the monotone convergence theorem, $\lim_{n \to \infty} v_n(p)$ exists for every $p \in \Delta(n)$, and is in fact equal to $v_\infty (p)$:
\begin{lemma}\label{lem-converge}
    The sequence $(v_n)_{n \in \mathbb{N}}$ pointwisely converges to $v_\infty$, that is
    \begin{equation} \label{eq-converge}
        \lim_{n \to \infty} v_n(p) = v_\infty(p), \ \forall p \in \Delta(\Omega).
    \end{equation}
\end{lemma}
\begin{proof}
    Any $n$-step sequential persuasion $S^{(n)}$ obtains the same expected utility as an infinite sequential persuasion $S$ that simulates $S^{(n)}$ for the first $n$ steps, and chooses trivial experiments afterward. Therefore $v_\infty(p) \ge v_n(p), \forall p$.
    
    For any $p \in \Delta(\Omega)$ and any $\epsilon > 0$, there exists an infinite sequential persuasion $S = (\Xi_0, \phi_1, \Xi_1, \ldots) \in \mathcal{S}_p$ with $\mathcal{V}(S) \ge v_\infty(p) - \frac{\epsilon}{2}$. By definition of $\mathcal{V}$, there exists $n \in \mathbb{N}$ s.t. 
    \[
        \sum_{\xi \in \Xi_{n}} [\mathbb{I}[S \text{ terminates after } \xi] \cdot \Pr[\xi] \cdot v(\text{last}(\xi))] \ge \mathcal{V}(S) - \frac{\epsilon}{2} \ge v_\infty(p) - \epsilon,
    \]
    therefore the $n$-step sequential persuasion $S^{(n)} \coloneqq (\Xi_0, \phi_1, \Xi_1, \ldots, \phi_n, \Xi_n) \in \mathcal{S}^{(n)}_p$ given by the first $n$ steps of $S$ satisfies $\mathcal{V}(S^{(n)}) \ge v_\infty(p) - \epsilon$. Since $\epsilon$ is arbitrary, we have $\sup_{n \in \mathbb{N}} v_n(p) \ge v_\infty (p)$.
\end{proof}

Unfortunately, computing $(v_n)_{n \ge 1}$ by the recursion (\ref{eq-iterate}) and then taking limit according to (\ref{eq-converge}) can be rather demanding for general $F$. Therefore we hope to obtain another characterization of $v_\infty$ which does not rely on the limit of any infinite sequence (note that the definition of $\mathcal{V}$ for infinite sequential persuasions also involves a limit). The characterization is given in Theorem \ref{thm-limit-characterize}. We define partial order ``$\ge$'' over the set of functions $g: \Delta(\Omega) \rightarrow \mathbb{R}$ by pointwise dominance: $g_1 \ge g_2$ if $g_1(p) \ge g_2(p)$ for every $p \in \Delta(\Omega)$.

\begin{theorem}\label{thm-limit-characterize}
    $v_\infty(p) = \inf\{g(p) ; g \in G\}$ for any $p \in \Delta(\Omega)$,
    where the set $G$ is given by
    \[
        G \coloneqq \left\{g: \Delta(\Omega) \rightarrow \mathbb{R};\  g \ge v, g(\sigma(e)) \ge \sum_{j=1}^m \lambda_j g(p_j) \text{ for every } e = (\lambda_j, p_j)_{j \in [m]} \in F\right\}.
    \]
\end{theorem}
If we replace $F$ in the definition of set $G$ by $F_0$, then 
$G$ becomes the set of concave functions that dominates $v$, and the corresponding $v_\infty$ becomes the concave envelope of $v$. For a smaller set $F \subset F_0$, the corresponding $G$ still contains all concave functions dominating $v$, therefore $v_\infty$ is less than or equal to the concave envelope of $v$.
\begin{proof}[Proof of Theorem \ref{thm-limit-characterize}]
    We first show $v_\infty$ itself is in $G$. Obviously $v_\infty \ge v$. Suppose there exists $e = (\lambda_j, p_j)_{j \in [m]} \in F$ s.t. $v_\infty(\sigma(e)) < \sum_{j=1}^m \lambda_j v_\infty(p_j)$. Denote $\epsilon \coloneqq \sum_{j=1}^m \lambda_j v_\infty(p_j) - v_\infty(\sigma(e))$, since $m$ is finite, and $(v_n(p_j))_{n \in \mathbb{N}}$ converges to $v_\infty(p_j)$ for any $j \in [m]$, there exists $n_0>0$ s.t. for any $n > n_0$ and any $j \in [m]$, $v_n(p_j) > v_\infty(p_j) - \epsilon$. Therefore, we have
    \[
        v_\infty(\sigma(e)) = \sum_{j=1}^m \lambda_j v_\infty(p_j) - \epsilon < \sum_{j=1}^m \lambda_j (v_n(p_j) + \epsilon) - \epsilon = \sum_{j=1}^m \lambda_j v_n(p_j) \le v_{n+1}(\sigma(e)),
    \]
    where the last inequality comes from Lemma \ref{lem-iterate}. This contradicts to Lemma \ref{lem-converge}, which tells that the increasing sequence $(v_n(\sigma(e)))_{n \in \mathbb{N}}$ converges to $v_\infty(\sigma(e))$.
    
    It suffices to show that $v_\infty$ is the minimum element of $G$ under partial order ``$\ge$''. Suppose $v_\infty$ is not the minimum, then there exists $g \in G$ s.t. $g(p) < v_\infty(p)$ for some $p \in \Delta(\Omega)$. By Lemma \ref{lem-converge}, for some $n \in \mathbb{N}$ we have $v_n(p) > g(p)$. Consider the minimum $n$ that makes this happen for some $(g, p)$:
    \[
        n^* \coloneqq \min\{n \in \mathbb{N} ; \text{there exists }g \in G, p \in \Delta(\Omega) \text{ such that } v_n(p) > g(p) \}.
    \]
    The set of which we take minimum is nonempty, discrete, and lower bounded by 1, therefore $n^* \ge 1$ exists. Consider some instances $g^* \in G$ and $p^* \in \Delta(\Omega)$ such that $v_{n^*}(p^*) > g^*(p^*)$. By Lemma \ref{lem-iterate}, there exists $(\lambda_j, p_j)_{j \in [m]} \in \mathcal{F}(p^*)$ s.t. 
    \begin{equation}\label{eq-contradict-instance}
        g^*(p^*) < \sum_{j=1}^m \lambda_j v_{n^*-1}(p_j).
    \end{equation}
    By definition of $n^*$ we have
    \begin{equation} \label{eq-last-valid-n}
        g \ge v_{n^*-1}, \forall g \in G.
    \end{equation}
    However, by combining equation (\ref{eq-last-valid-n}) with the definition of $G$, we can also obtain 
    \[
        g^*(p^*) \ge \sum_{j=1}^m \lambda_j g^*(p_j) \ge \sum_{j=1}^m \lambda_j v_{n^*-1}(p_j),
    \]
    which contradicts to equation (\ref{eq-contradict-instance}). We conclude that $v_\infty$ is indeed the minimum element of $G$, and this finishes the proof of the theorem.
\end{proof}

The representation of $v_\infty$ in Theorem \ref{thm-limit-characterize} leads to a useful observation, that any $g \in G$ can be used as a certificate that limits $v_\infty$ from above. The rest of the section is some quick applications of this property.

\subsection{When finite steps are sufficient}
In reality, it is natural that the number of usable experiment steps is limited, or that each experiment has a cost so that the sender has to handle the trade-off between extra costs and the increase in utility by adding more steps. Here we focus on a degenerate case, that is, when it suffices to take only a finite number of steps. By Theorem \ref{thm-limit-characterize} and Lemma \ref{lem-iterate}, such case is marked by some corresponding $g \in G$:
\begin{coro}
    For any $n \in \mathbb{N}$ and $p \in \Delta(\Omega)$, $v_n(p) = v_\infty(p)$ if and only if there exists $g: \Delta(\Omega) \rightarrow \mathbb{R}$ such that $g \ge v$, $g(\sigma(e)) \ge \sum_{j=1}^m \lambda_j g(p_j)$ for every $e = (\lambda_j, p_j)_{j \in [m]} \in F$, and $g(p) = v_n(p)$.
\end{coro}
In general, there can be infinitely many $g \in G$ satisfying the corollary's requirement, but the minimum element $v_\infty$ is unique. Therefore, finding one such $g$ can be much easier than finding $v_\infty$.

\subsection{Parameterized estimation}
Every $g \in G$ bounds $v_\infty$ from the above, and by choosing a sufficiently representative $G' \subseteq G$, we can possibly obtain nice estimations of $v_\infty$. One idea is to take $G'$ to be a parameterized function family, so that (for $F$ with some nice structure) the condition $g(\sigma(e)) \ge \sum_{j=1}^m \lambda_j g(p_j)$ can be quickly verified throughout $F$.  

For example, consider $\Omega = \{(\theta_0, \theta_1) \in \{0,1\}^2\}$,
and every feasible experiment is a test for one of $\theta_0$ and $\theta_1$ that is independent of the other variable, that is, every $e \in F$ only changes the marginal distribution of $\theta_d$ for some $d \in \{0,1\}$, while the marginal distribution of $\theta_{1-d}$ for every belief after $e$ are unchanged.
In this case, $\Delta(\Omega)$ can be represented by the unit square whose $x$ and $y$ axis correspond to the probabilities of $\theta_0 = 1$ and $\theta_1 = 1$ separately, and every experiment in $F$ spreads the belief either horizontally or vertically. For any such $F$, the corresponding $G$ contains the set $G'$ of bilinear forms that pointwisely dominate $v$, which belongs to a function family of four parameters:
\begin{align*}
    G' \coloneqq \big \{& g: \Delta(\Omega) \rightarrow \mathbb{R}; \ \  g \ge v, c_1, c_2, c_3, c_4 \in \mathbb{R}, \\
    & g((p^{(0)}_j, p^{(1)}_j)) = c_1  p^{(0)}_j p^{(1)}_j + c_2 p^{(0)}_j + c_3 p^{(1)}_j + c_4 \text{ for every } (p^{(0)}_j, p^{(1)}_j) \in [0,1]^2 \big \}.
\end{align*}

An immediate question is whether such parameterized families can be used to obtain the exact $v_\infty$, like in the classic Bayesian persuasion where the concave closure can be obtained by considering only linear functions (see e.g. \citet{duality}). Unfortunately, there exists $(\mu, F, v)$ for this example where the optimal expected utility depends on the value of $v$ at arbitrarily many points in $[0,1]^2$, which rules out the possibility of finding $v_\infty$ with a finite parameter function family.

\section{Implementation of optimal Bayesian persuasion}
Before analyzing the general existence and properties of optimal sequential persuasions, we first study a special case, that is, when the optimal expected utility by the classic Bayesian persuasion with no limit on feasible experiments, which by \citet{bayesian-persuasion} is equal to the concave closure of $v$, can be achieved by some sequential persuasion using experiments in $F$. In this section, we introduce extra assumptions about the feasible set $F$, so that accessible criteria for this special case can be established. 

Since we assume $v$ to be upper semicontinuous and $\Omega$ to be finite, the optimal utility of Bayesian persuasion is always attainable by some distribution over $\Delta(\Omega)$ whose support size is at most $|\Omega|$, which is in $F_0$. Therefore, it suffices to use the maximum over $F_0$ to define the concave closure $\hat v: \Delta(\Omega) \rightarrow \mathbb{R}$, which is the mapping from the prior to the optimal sender's utility of classic Bayesian persuasion:
\[
    \hat v(p) \coloneqq \max \left\{\sum_{j = 1}^m \lambda_j v(p_j); e= (\lambda_j, p_j)_{j \in [m]} \in F_0, \sigma(e) = p \right\}, \forall p \in \Delta(\Omega).
\]
By the discussion after Theorem \ref{thm-limit-characterize}, we always have $\hat v \ge v_\infty$. 

A nice property of the optimal experiments in $F_0$ is that each of them can be fully identified by its support, and the same optimal support is shared by all priors within its convex hull. Denote by $\text{conv}(O)$ the convex hull of a set $O \subseteq \Delta(\Omega)$, we rigorously state this property as the following lemma.
\begin{lemma}\label{lem-support-indentify}
    For any $(v, \mu)$, there exists some $O \subseteq \Delta(\Omega)$ with $\mu \in \text{conv}(O)$ such that, for any $p \in \text{conv}(O)$ and any $e = (\lambda_j, p_j)_{j \in [m]} \in F_0$ with $\sigma(e) = p$, $\sum_{j=1}^m \lambda_j v(p_j) = \hat v(p)$ if and only if $\tau(e) \subseteq O$.
\end{lemma}
\begin{proof}
    For every $v$, $\hat v$ is given by the concave closure of $v$, therefore for every $\mu \in \Delta(\Omega)$ there exists an affine function $f_\mu$ defined on $\Delta(\Omega)$ so that $f_\mu(\mu) = \hat v(\mu)$, $f_\mu(p) \ge \hat v(p) \ge v(p)$ for every $p \in \Delta(\Omega)$, and $f_\mu(p') = \hat v(p') = v(p')$ for every optimal information structure $e_0 \in F_0$ and every $p' \in \tau(e)$. Now consider the set
    \[
        O \coloneqq \{p \in \Delta(\Omega); f_\mu(p) = v(p)\}.
    \]
    By the affine property of $f_\mu$ and concavity of $\hat v$, for any $p \in \text{conv}(O)$ we have $f_\mu(p) = \hat v(p)$. Moreover, for any $e = (\lambda_j, p_j)_{j \in [m]}$ with $\tau(e) \in O$ and $\sigma(e) \in \text{conv}(O)$, the affine property of $f_\mu$ gives
    \[
        \sum_{j=1}^m \lambda_j v(p_j) 
        = \sum_{j=1}^m \lambda_j f_\mu(p_j) 
        = f_\mu(\sigma(e)) 
        = \hat v(p).
    \]
    Therefore $e$ is an optimal experiment that maximizes the expected utility over all mean-preserving spreads of $p$.
\end{proof}

By Lemma \ref{lem-support-indentify}, the sender achieves the optimal expected utility of the classic Bayesian persuasion if and only if the receiver always obtains some belief in $O$, which is irrelevant to the exact distribution of belief over $O$ (e.g. when $|O| > |\Omega|$, there can be infinite many $e \in F_0$ with support in $O$ and $\sigma(e) = \mu$). When we take sequential persuasions into consideration, the following criterion is useful for characterizing the existence of an infinite sequential persuasion that attains the optimal expected utility of the classic Bayesian persuasion. 
\begin{define}\label{define-implement}
    For any instance $(\mu, F)$, set $O \subseteq \Delta(\Omega)$ is implementable if for any $\epsilon > 0$, there exists $n \in \mathbb{N}$ and $S^{(n)} \in \mathcal{S}^{(n)}_\mu$ such that the random belief $b_n$ induced by $S^{(n)}$ satisfies $Pr[b_n \in O] > 1-\epsilon$.
\end{define}
Note that in this definition, only the probability that the finite-step belief comes exactly into $O$ is considered, and whether the belief is in some (arbitrarily small) neighborhood of $O$ does not matter.
An illustration of this is the example in Section \ref{sec-terminate-utility}. The same example tells that whether $v_\infty(\mu) = \hat v(\mu)$, which is a weaker property than the existence of $S \in \mathcal{S}_\mu$ with $\mathcal{V}(S) = \hat v(\mu)$, cannot be characterized by the probability that the finite-step belief comes into $O$. 



Function $v_\infty$ can be used to derive a useful characterization of an arbitrary set $O$'s implementability, which is stated as the following lemma.
\begin{lemma}\label{lem-implement}
    For any $(\mu, F)$, set $O\subseteq \Delta(\Omega)$ is implementable if and only if the instance $(\mu, F, v)$ with
    \[
        v(p) \coloneqq \left \{
        \begin{aligned}
        1 \quad& p \in O \\
        0 \quad & p \in \Delta(\Omega) \backslash O
        \end{aligned}\right.
    \]
    gives $v_\infty(\mu) = 1$.
\end{lemma}
\begin{proof}
    \textbf{If part:} Suppose $v_\infty(\mu) = 1$, then for any $\epsilon>0$, there exists $n \ge 0$ s.t. $v_n(\mu) > 1 - \frac{\epsilon}{2}$, which further means there exists $S^{(n)} \in \mathcal{S}^{(n)}_\mu$ such that 
    \[
        \mathcal{V}(S^{(n)}) \ge v_n(\mu) - \frac{\epsilon}{2} = 1-\epsilon,
    \]
    therefore with probability $1-\epsilon$ the belief $b_n$ induced by $S^{(n)}$ is in $O$. Since $\epsilon>0$ is arbitrary, $O$ is implementable.

    \noindent
    \textbf{Only if part:} Suppose $O$ is implementable, then for any $\epsilon > 0$, there exists some finite-step sequential persuasion which starts from $\mu$, and ends with belief in $O$ with probability at least $1-\epsilon$. Therefore $v_\infty(\mu) \ge 1-\epsilon$ for any $\epsilon > 0$, whence $v_\infty(\mu) = 1$. 
\end{proof}

\subsection{Additional assumptions}
So far we have discussed the utility structure of sequential persuasions. In order to obtain more structural results, we need to make additional assumptions. The first assumption follows the intuition that an experiment itself is feasible if it can be approximated arbitrarily well using feasible experiments.
\begin{assume}\label{assume-closed-F}
    There exists $h \in \mathbb{N}$ s.t. $|\tau(e)| \le h, \forall e \in F$. Moreover, for any sequence $(e_i)_{i \ge 1}$ in $F$ that weakly converges to $e \in \Delta(\Delta(\Omega))$, $e$ is also in $F$.
\end{assume}
It is easy to construct $F$ that does not contain the weak limit of experiment sequence $(e_i)_{i \ge 1}$ with $e_i\in F$. In this case, for some corresponding $v$, the values of $(v_n)_{n \in \mathbb{N}}$ and $v_\infty$ cannot be exactly attained by any corresponding finite-step and infinite sequential persuasions. The upper bound $h$ on support size guarantees the validity of the claim that $F$ contains all weak limits: If the elements in $F$ have unbounded support size, although by Prokhorov's theorem the weak limit $e \in \Delta(\Delta(\Omega))$ always exists, it is not necessarily in $F_0$, which may break $F \subseteq F_0$. For example, consider the sequence of uniform distributions on the sequence of sets $(\{\frac{j}{i}; 1 \le j \le i\})_{i \ge 1}$, every distribution in the sequence has a finite support. However, its weak limit is the uniform distribution over $[0, 1]$ interval. With the upper bound $h$, it can be easily shown that the weak limit $e$ also satisfies $|\tau(e)| \le h$.

Assumption \ref{assume-closed-F} is not sufficient for the existence of infinite sequential persuasion that attains the value of $v_\infty$. The following assumption is also considered. 
\begin{assume}\label{assume-entropy}
    There exists $\delta > 0$ s.t. for any $e = (\lambda_j, p_j)_{j \in [m]} \in F \backslash T$, $\sum_{j \in [m]} \lambda_j H(p_j) \le H(\sigma(e)) - \delta$.
\end{assume}
Here $H: \Delta(\Omega) \rightarrow \mathbb{R}_{\ge 0}$ is the entropy function: $H(p) \coloneqq - \sum_{\omega \in \Omega} p(\omega) \log p(\omega)$. 
A motivation behind Assumption \ref{assume-entropy} is as follows: Suppose that, to correctly interpret of a nontrivial experiment $e \in F \setminus T$ and calculate the corresponding Bayesian update, the receiver has to incur some cost $c > 0$. The receiver's utility function $u: \Delta(\Omega) \rightarrow \mathbb{R}$ is the pointwise maximum of the bounded linear functions $f_a(p) \coloneqq \mathbb{E}_{\omega \sim p} U(\omega, a)$ for all $a \in A$, therefore $u$ is convex and bounded. If $e = (\lambda_j, p_j)_{j \in [m]}$ reduces the expected entropy of the belief by only a very small amount, then with high probability, the receiver's new belief is very similar to the belief before $e$. In this case, by interpreting $e$ and updating its belief, the receiver's increase in utility $[\sum_{j \in [m]}\lambda_j u(p_j) - u(\sigma(e))]$ does not exceed the incurred cost $c$. To incentivize the receiver to update its belief as expected, the sender has to ensure that every nontrivial experiment reduces the expected entropy of belief by some non-negligible amount.

Although being easy to interpret, Assumption \ref{assume-entropy} is too strong for our purpose. We find the following assumption to be a necessary condition for Assumption \ref{assume-entropy}, and we show in Theorem \ref{thm-opt-exist} that under Assumption \ref{assume-closed-F}, it is a necessary and sufficient condition for $v_\infty(\mu)$ to be attainable by some $S \in \mathcal{S}_\mu$.
\begin{assume}\label{assume-uniform-as}
    For prior $\mu$, there exists a sequence of finite-step sequential persuasions $(S_k^{(n_k)})_{k \ge 1}$ starting from $\mu$ such that:
    \begin{itemize}
        \item $\lim_{k \to \infty} \mathcal{V}(S_k^{(n_k)}) = v_\infty(\mu)$, 
        \item for every $\epsilon > 0$, there exists $n_\epsilon \in \mathbb{N}$ such that for every $k$ with $n_k \ge n_\epsilon$,
        \[
            \Pr[S_k^{(n_k)} \text{ terminates after } n_\epsilon \text{ steps}] \ge 1-\epsilon.
        \]
    \end{itemize}
\end{assume}
Note that there always exists $(S_k^{(n_k)})_{k \ge 1}$ whose limiting utility is $v_\infty(\mu)$, and every finite-step sequential persuasion definitely terminates after sufficiently many steps. Assumption \ref{assume-uniform-as} requires that such a sequence shares a minimum rate of termination, that is, the same pair $(\epsilon, n_\epsilon)$ can be used to lower-bound the termination probability of every element in the sequence.

In fact, Assumption \ref{assume-entropy} implies Assumption \ref{assume-uniform-as} in a ``rough'' way, which is to directly limit the connectivity over $\Delta(\Omega)$ that $F$ creates, so that any sequential persuasion has to terminate with high probability after some fixed steps regardless of $\mu$ and the exact choice of experiments: For every prior $\mu \in \Delta(\Omega)$ we have $H(\mu) \le \log(\Omega)$. Assumption \ref{assume-entropy} requires that taking each nontrivial experiment reduces entropy by at least $\delta$, therefore the probability of taking more than $n$ nontrivial experiments is upper bounded by $\log(|\Omega|) / (n\delta)$. By taking a sequence of finite-step sequential persuasions whose expected utility goes to $v_\infty(\mu)$ and neglecting all nontrivial experiment taken before termination, we can obtain a sequence in which every term terminates after $n$ steps with probability at least $1-\log(|\Omega|) / (n\delta)$.

\subsection{Implementability and structure of implementations}
In previous sections, we define the set of feasible experiments $F$ to be an arbitrary subset of $F_0$. Such an arbitrary $F \subseteq F_0$ can be very irregular, so that although $v$ is upper semicontinuous, $(v_n)_{n \ge 1}$ can easily lose upper semicontinuity unless extra assumptions are introduced. The case with the limiting function $v_{\infty}$ can be even more complex, as we show in Section \ref{sec-triangle}. We use the following lemma to guarantee the upper semicontinuity of these functions.
\begin{lemma}\label{lem-upper-semicont}
    Under Assumption \ref{assume-closed-F}, for every $n \ge 0$, $v_n$ is upper semicontinuous. Under Assumption \ref{assume-closed-F} and \ref{assume-entropy}, $v_\infty$ is upper semicontinuous.  
\end{lemma}

With upper semicontinuous $(v_n)_{n \in \mathbb{N}}$ and $v_{\infty}$, the following theorem concludes our findings about those infinite sequential persuasions that implement the optimal Bayesian persuasion: They exist as long as $O$ is implementable, and can be traced by a set of intermediate states $D \subseteq \Delta(\Omega)$, in which every belief can be spread by a nontrivial persuasion which never lets the belief escape $D \cup O$.
\begin{theorem}\label{thm-implement}
    For any $(\mu, F, v)$ satisfying Assumption \ref{assume-closed-F} and \ref{assume-entropy}, the following statements are equivalent:
    \begin{enumerate}
        \item Set $O \subseteq \Delta(\Omega)$ given by Lemma \ref{lem-support-indentify} is implementable.
        \item There exists some $D \subseteq \Delta (\Omega)$ and some $F_D \subseteq F \backslash T$ such that (i) $\mu \in D \cup O$, (ii) $F_D \cap \mathcal{F}(p) \ne \emptyset$ for every $p \in D$, and (iii) $\cup_{e \in F_D} \tau(e) \subseteq D \cup O$.
        \item There exists $S \in \mathcal{S}_\mu$ such that $\mathcal{V}(S) = \hat v(\mu)$.
    \end{enumerate}
\end{theorem}
\begin{proof}[Proof of Theorem \ref{thm-implement}]
    (1) $\Rightarrow$ (2): Suppose $O$ is implementable, consider function $v$ given in Lemma \ref{lem-implement}, define set $D$ by
    \[
        D \coloneqq \{p \in \Delta(\omega); v_\infty(p) = 1\} \setminus O,
    \]
    we hope to show that $D$ and some $F_D$ satisfies (i), (ii) and (iii). By Lemma \ref{lem-implement} we have $\mu \in D \cup O$ so that $D$ satisfies (i). 
    
    For any $p \in D$ and any arbitrary positive decreasing sequence $(\epsilon_i)_{i \ge 1}$ that goes to zero, by Theorem \ref{thm-limit-characterize}, there exists a sequence of nontrivial experiments $(e_i)_{i\ge 1}$ in $\mathcal{F}(p)$ in which every $e_i=(\lambda^{(i)}_j, p^{(i)}_j)_{j \in [m_i]}$ satisfies $\sum_{j \in [m_i]} \lambda^{(i)}_j v_\infty(p^{(i)}_j) \ge 1-\epsilon_i$. With the compactness of $\Delta(\Delta(\Omega))$ and Assumption \ref{assume-closed-F}, there exists a subsequence of $(e_i)_{i \ge 1}$ indexed by $(i_k)_{k \ge 1}$ that weakly converges to some $e^* = (\lambda_j^*, p_j^*)_{j \le m^*} \in F$. Since every $e_{i_k}$ is in $\mathcal{F}(p)$, we have $\sigma(e^*) = p$. Moreover, by Assumption \ref{assume-entropy}, the amount of information revealed by every $e_{i_k}$ is at least $\delta$, therefore $e^*$ is not a trivial experiment.
    
    By Lemma \ref{lem-upper-semicont}, $v_\infty$ is bounded and upper semicontinuous, therefore for the weak limit $e^*$ of $(e_{i_k})_k \ge 1$ we have 
    \begin{align*}
        \sum_{j=1}^{m^*} \lambda_j^* v_\infty(p_j^*) \ge \lim_{k \to \infty} \sum_{j \in [m_{i_k}]]} \lambda^{(i_k)}_j v_\infty(p^{(i_k)}_j)
        \ge \lim_{k \to \infty} 1 - \epsilon_{i_k} 
        = 1.
    \end{align*}
    Therefore $v_\infty(p_j^*) = 1, \forall j \in [m^*]$. We can add $e^*$ to $F_D$ as it satisfies $\tau(e^*) \subseteq D \cup O$. (ii) and (iii) can be satisfied by considering $F_D$ that consists of the corresponding $e^*$ for every element of $D$.
    

    \noindent
    (2) $\Rightarrow$ (3): Suppose there exists $(D, F_D)$ satisfying the requirement of (2). Then we can define an infinite sequential persuasion $S$ starting from $\mu$ as follows. If the current belief is in $D$, take an arbitrary experiment in $F_D$ that is a mean-preserving spread of the current belief, which exists because of (ii). If the current belief is not in $D$, take a trivial experiment. The discussion after Assumption \ref{assume-entropy} tells that $S$ terminates after $n$ steps with probability at least $1-\frac{\log(|\Omega|)}{n\delta}$. For any finite $n$, if $S$ terminates after $n$ steps with positive probability at some specific belief $p \in \Delta(\Omega)$, then by (iii) $p$ has to be in $O$. Therefore, for any $\epsilon > 0$, there exists $n_\epsilon \ge 1$ such that $S$ terminates after $n_\epsilon$ steps at belief in $O$ with probability $1-\epsilon$. Denote by $p_n$ the expectation of $n$-step belief conditioned on termination: 
    \[
        p_n \coloneqq \frac{\sum_{\xi \in \Xi_n} \mathbb{I}[S \text{ terminates after } \xi] \cdot \Pr[\xi] \cdot \text{last}(\xi)}{\sum_{\xi \in \Xi_n} \mathbb{I}[S \text{ terminates after } \xi] \cdot \Pr[\xi]}. 
    \]
    Then every $p_n$ is in $\text{conv}(O)$, by Lemma \ref{lem-support-indentify} we have
    \[
        \hat v(p_n) = \frac{\sum_{\xi \in \Xi_n} \mathbb{I}[S \text{ terminates after } \xi] \cdot \Pr[\xi] \cdot v(\text{last}(\xi))}{\sum_{\xi \in \Xi_n} \mathbb{I}[S \text{ terminates after } \xi] \cdot \Pr[\xi]}.
    \]
    Therefore, for every $\epsilon > 0$ we have
    $\mathcal{V}(S) \ge (1-\epsilon) \cdot \hat v(p_{n_\epsilon})$.
    It can be easily established that $p_{n_\epsilon}$ converges to $\mu$ as $\epsilon$ goes to zero. By the continuity of $\hat v$, we have
    \[
        \mathcal{V}(S) \ge \lim_{\epsilon \to 0} (1-\epsilon) \cdot \hat v(p_{n_\epsilon}) = \hat v(\mu).
    \]

    \noindent
    (3) $\Rightarrow$ (1): As shown in the discussion after Assumption \ref{assume-uniform-as}, for any $S \in \mathcal{S}_\mu$ with $\mathcal{V}(S) = \hat v(\mu)$ and any $\epsilon > 0$, there exists $n_\epsilon \ge 1$ such that $S$ terminates after $n_\epsilon$ steps with probability at least $1-\epsilon$, and since $n_\epsilon$ is finite, there is only a finite number of possible beliefs that $S$ terminates at after $n_\epsilon$ steps with positive probability. Suppose there exists some belief $p \notin O$ so that $S$ terminates after $n$ steps at $p$ with positive probability, then it is easy to obtain $\mathcal{V}(S) < \hat v(\mu)$, a contradiction. Therefore, for every $\epsilon > 0$, the corresponding $n_\epsilon$ and $S^{(n_\epsilon)}$ given by the first $n_\epsilon$ steps of $S$ satisfy the requirement of Definition \ref{define-implement}, and $O$ is implementable.
\end{proof}

\subsection{An infinite-step example}\label{sec-triangle}
The necessity of Assumption \ref{assume-entropy} (or Assumption \ref{assume-uniform-as} in general) can be shown by such an example, which is plotted as Figure \ref{fig-c-necessary}: Consider $\Omega = \{1,2,3\}$ so that each belief can be represented by triplet $(x,y, 1-x-y)$ and plotted in a triangle whose vertices are $A_0 = (0,0,1), B_0 = (0,1,0), C_0 = (1,0,0)$. Function $v$ is given by $v(p) \coloneqq \mathbb{I}[p \in \{A_0, B_0, C_0\}]$, therefore only takes nonzero value one the three vertices of the triangle, which we mark using black dots. The prior is the center of the triangle $\mu = \{\frac13, \frac13, \frac13\}$, and obviously $\hat v(\mu) = 1$. We are going to construct two different $F$ that both satisfy Assumption \ref{assume-closed-F}, while the first one violates Lemma \ref{lem-upper-semicont} and the second one satisfies Lemma \ref{lem-upper-semicont} but violates Theorem \ref{thm-implement}.

To construct the first $F$, for $i = 1, 2, \ldots$, iteratively define beliefs $A_i, B_i, C_i$ by $A_i = \frac12 A_{i-1} + \frac12 B_{i-1}, B_i = \frac12 B_{i-1} + \frac12 C_{i-1}, C_i = \frac12 C_{i-1} + \frac12 A_{i-1}$, and correspondingly construct a set of three nontrivial experiments $F_i = \{e_i^{(1)} = ((\frac12, A_{i-1}), (\frac12, B_{i-1})), e_i^{(2)} = ((\frac12, B_{i-1}), (\frac12, C_{i-1})), e_i^{(3)} = ((\frac12, C_{i-1}), (\frac12, A_{i-1}))\}$. Then take $F \coloneqq (\cup_{i \ge 1} F_i) \cup T$, where $T$ is the set of all trivial experiments. The support size of every element in $F$ is at most 2. Every infinite convergent sequence in $F$ contains at least one infinite subsequence of $(e_i^{(k)})_{i \ge 1}$ for some $k \in \{1,2,3\}$ or an infinite convergent sequence in $T$. If it contains an infinite subsequence of $(e_i^{(k)})_{i \ge 1}$, then the limiting experiment is $(1, \mu)$, which is in $F$. Otherwise the limiting experiment is also a trivial experiment, hence also in $F$. By induction, one can easily check $v_\infty(A_i) = v_\infty(B_i) = v_\infty(C_i) = 1, \forall i \ge 1$. Moreover, $(A_i)_{i \ge 1}$ is a sequence in $\Delta(\Omega)$ that converges to $\mu$. However, $F$ does not contain any nontrivial experiment that is a mean-preserving spread of $\mu$, therefore $v_\infty(\mu) = 0$. This shows that $v_\infty$ is not upper semicontinuous.

The second $F$ is constructed by adding a sequence of experiments to the first one: Take a fixed belief $W \coloneqq (\frac5{12}, \frac16, \frac1{52})$. For every $i = 0, 1, \ldots$, note that belief $D_i \coloneqq (\frac13 - \frac{1}{3\cdot 4^i}, \frac13 + \frac{2}{3\cdot 4^i}, \frac13 - \frac{1}{3\cdot 4^i})$ is one of $A_i$, $B_i$ and $C_i$. We also add experiments $e_i^{(4)} = ((\frac{1}{1+4^{i-1}}, W), (\frac{4^{i-1}}{1+4^{i-1}}, D_i))$ to $F$. One can easily check that after adding $e_i^{(4)}$ to $F$, Assumption \ref{assume-closed-F} still holds. Every $e_i^{(4)}$ is a mean-preserving spread of $\mu$ into $W$ with $v(W) = v_\infty(W) = 0$ and $D_i$ with $v_\infty(D_i) = 1$, and the probability that the outcome of $e_i^{(4)}$ changes belief to $D_i$ goes to 1 as $i$ goes to infinity, which gives $v_\infty(\mu) = 1$. However, every sequential persuasion starting from $\mu$ has to choose some $e_i^{(4)}$ as its first nontrivial experiment, which always gives belief $W$ some positive probability. Therefore, the expected utility of every finite-step or infinite sequential persuasion is strictly smaller than 1.

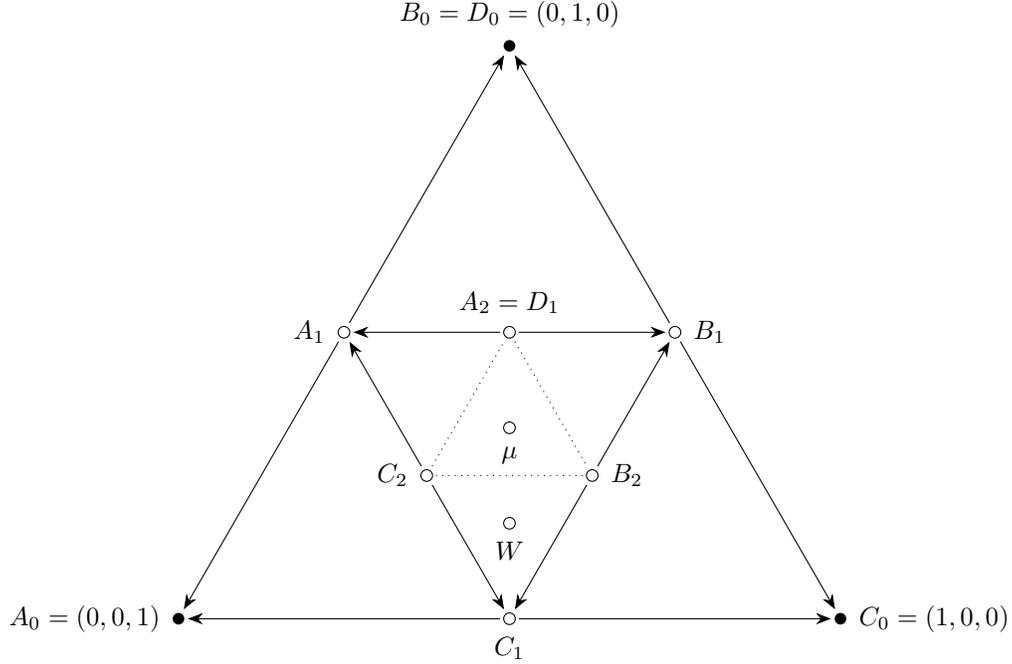
\begin{figure}
    \centering
    
    \begin{tikzpicture}[scale = 2.2, >={Stealth[scale = 1.2]}]
        \path (0,0) node[label=left:{$A_0=(0,0,1)$}] (A0){}
        -- ++(60:4) node[label=above:{$B_0=D_0 = (0,1,0)$}] (B0){}
        -- ++(-60:4) node[label=right:{$C_0=(1,0,0)$}] (C0){};
        \fill (A0) circle(1pt);
        \fill (B0) circle(1pt);
        \fill (C0) circle(1pt);
        \path (60:2) node[label=left:{$A_1$}] (A1){}
        -- ++(0:2) node[label=right:{$B_1$}] (B1){}
        -- ++(-120:2) node[label=below:{$C_1$}] (C1){};
        \draw (A1) circle(1pt);
        \draw (B1) circle(1pt);
        \draw (C1) circle(1pt);
        \draw[->] (A1) -- (A0);
        \draw[->] (A1) -- (B0);
        \draw[->] (B1) -- (B0);
        \draw[->] (B1) -- (C0);
        \draw[->] (C1) -- (C0);
        \draw[->] (C1) -- (A0);
        \path (60:2) --
        ++ (0:1) node[label=above:{$A_2 = D_1$}] (A2){}
        -- ++(-60:1) node[label=right:{$B_2$}] (B2){}
        -- ++(180:1) node[label=left:{$C_2$}] (C2){};
        \draw (A2) circle(1pt);
        \draw (B2) circle(1pt);
        \draw (C2) circle(1pt);
        \draw[->] (A2) -- (A1);
        \draw[->] (A2) -- (B1);
        \draw[->] (B2) -- (B1);
        \draw[->] (B2) -- (C1);
        \draw[->] (C2) -- (C1);
        \draw[->] (C2) -- (A1);
        \draw (30:2.3094) circle(1pt) node [label = below:{$\mu$}]  (mu){};
        \draw[dotted] (A2) -- (B2) -- (C2) -- (A2);
        \path (0:2) -- ++(90:0.57735) node [label = below:{$W$}] (W){};
        \draw (W) circle(1pt);
    \end{tikzpicture}
    
    \caption{An example that violates Assumption \ref{assume-entropy}}
    \label{fig-c-necessary}
\end{figure}

\section{Optimal Markov sequential persuasion}
In this section we focus on the existence of optimal sequential persuasions. An optimal sequential persuasion is defined to be a sequential persuasion whose expected utility is equal to the corresponding supremum given by function $v_n$ or $v_\infty$:
\begin{define}
    For every $n \in \mathbb{N}$, $S^{(n)} \in \mathcal{S}^{(n)}_\mu$ is an optimal $n$-step sequential persuasion starting from $\mu$ if $\mathcal{V}(S^{(n)}) = v_n(\mu)$. $S \in \mathcal{S}_\mu$ is an optimal infinite sequential persuasion starting from $\mu$ if $\mathcal{V}(S) = v_\infty(\mu)$.
\end{define}
Under Assumption \ref{assume-closed-F}, by Lemma \ref{lem-upper-semicont} $v_n$ is upper semicontinuous for all $n \in \mathbb{N}$. Therefore, every supremum in equation (\ref{eq-iterate}) can be attained by some $e \in \mathcal{F}_p$, from which one can easy derive the existence of optimal $n$-step sequential persuasions.
\begin{coro}
    Under Assumption \ref{assume-closed-F}, for any $\mu \in \Delta(\Omega)$, there exists an optimal $n$-step sequential persuasion starting from $\mu$.
\end{coro}

For infinite sequential persuasions, we are going to derive the following results: First, under Assumption \ref{assume-closed-F}, an optimal infinite sequential persuasion exists if and only if $(\mu, F, v)$ also satisfies Assumption \ref{assume-uniform-as}. 
Moreover, as mentioned in Section \ref{sec-discuss-markov}, when an optimal infinite sequential persuasion exists, there actually exists some optimal infinite sequential persuasion which is \emph{Markov}, that is, the choice of the next experiment only depends on the receiver's current belief. 
\begin{define}\label{define-markov}
    For any tuple $(\mu, D, Z, \rho)$ where:
    \begin{enumerate}
        \item $D \subseteq \Delta(\Omega)$ is a the set of intermediate beliefs,
        \item $Z \subseteq \Delta(\Omega) \backslash D$ is a the set of terminating beliefs, $\mu \in D \cup Z$,
        \item $\rho: D \rightarrow F \backslash T$ that determines the next experiment with $\sigma(\rho(p)) = p$ and $\tau(\rho(p)) \subseteq D \cup Z$, for all $p \in D$,
    \end{enumerate}
    an infinite sequential persuasion $S = (\Xi_0, \phi_1, \Xi_1, \ldots)$ starting from $\mu$ is a Markov sequential persuasion compatible with $(\mu, D, Z, \rho)$ if for every $i \in \mathbb{N}$ and every $\xi \in \Xi_i$, 
    \begin{itemize}
        \item if $\text{last}(\xi) \in D$, then $\phi_i(\xi) = \rho(\text{last}(\xi))$;
        \item if $\text{last}(\xi) \in Z$, then $\phi_i(\xi) = (1, \text{last}(\xi))$.
    \end{itemize}
    An infinite sequential persuasion $S$ starting from $\mu$ is an optimal Markov sequential persuasion if it is a Markov sequential persuasion compatible with some tuple $(\mu, D, Z, \rho)$, and is an optimal infinite sequential persuasion.
\end{define}

Recall the second $F$ we construct for the example of Figure \ref{fig-c-necessary}. No optimal infinite sequential persuasion exists for the $(\mu, F, v)$ because no matter which experiment is used as the first step to change belief from $\mu$, the new belief always comes to $W$ with positive probability, which causes a loss of the expected utility, making it lower than the supremum $v_\infty(\mu)$. In this sense, every feasible experiments that can change $\mu$ is ``imperfect''. This motivates the following definition of ``perfect'' experiments in $F$.
\begin{define}
    Given $(F, v)$, experiment $e = (\lambda_j, p_j)_{j \in [m]} \in F$ is an exact experiment if
    \[
        v_\infty(\sigma(e)) = \sum_{j \in [m]} \lambda_j v_\infty(p_j).
    \]
    Also, denote by $N$ the set where $v_\infty$ and $v$ coincide:
    \[
        N \coloneqq \{p \in \Delta(\Omega); v_\infty(p) = v(p)\}.
    \]
\end{define}
\begin{lemma}\label{lem-Markov-optimal}
    A Markov sequential persuasion $S \in \mathcal{S}_\mu$ compatible with $(\mu,D,Z,\rho)$ is optimal if $Z \subseteq N$, $\rho$ only takes exact experiments, and the for every $\epsilon > 0$, there exists $n \in \mathbb{N}$ such that the $n$-step belief $b_n$ induced by $S$ satisfies $\Pr[b_n \in Z] \ge 1-\epsilon$.
\end{lemma}
To show the existence of an optimal infinite sequential persuasion, we need a more structured representation of infinite sequential persuasions, which we define as follows. The basic structure is a tree of infinite depth, so that the $n$-step histories are encoded by vertices of depth $n$, and the topology of the tree tells how the history grows as more experiments are taken. With Assumption \ref{assume-closed-F}, every experiment generates at most $h$ different beliefs, and we require that every vertex in the tree has exactly $h$ children.
\begin{define}\label{def-h-branch}
    An $h$-branching sequential persuasion $B \coloneqq (C_n, \pi_n, \beta_n, \eta_n)_{n \in \mathbb{N}}$ starting from $\mu$ is defined on an infinite-depth tree where every vertex has $h$ children, and for every $n \in \mathbb{N}$,
    \begin{itemize}
        \item $C_n \coloneqq \{c_{n, m} | 1 \le m \le h^n\}$ is the set of $h^n$ vertices of depth $n$ in the tree, the set of children of $c_{n,m}$ being $\Gamma(n, m) \coloneqq \{c_{n+1, (m-1)h+j} | 1 \le j \le h\}$;
        \item a probability distribution $\pi_n$ over $C_n$;
        \item two mappings $\beta_n: C_n \rightarrow \Delta(\Omega)$ and $\eta_n: C_n \rightarrow F$,
    \end{itemize}
    such that:
    \begin{itemize}
        \item Belief transitions are compatible: $\beta_0(c_{0,1}) = \mu$, $\sigma(\eta_n(c_{n, m})) = \beta_n(c_{n, m})$ and $\tau(\eta_n(c_{n, m})) \subseteq \{\beta_{n+1}(c) | c \in \Gamma(n, m)\}$ for every $n \in \mathbb{N}$ and $1 \le m \le h^n$.
        \item Probability distributions are consistent: For every $n \in \mathbb{N}$ and $1 \le m \le h^n$, for every $p$ assigned with probability $\lambda \ge 0$ by $\eta_n(c_{n, m})$, $\sum_{c \in \Gamma(n, m)} \mathbb{I}[\beta_{n+1}(c) = p] \cdot \pi_{n+1}(c) = \pi_n(c_{n, m}) \cdot \lambda$.
    \end{itemize}
    And $B = (C_n, \pi_n, \beta_n, \eta_n)_{n \in \mathbb{N}}$ terminates after $c_{n,m}$ if $\eta_n(c_{n,m}) \in T$ and every descendant $c_{n',m'}$ of $c_{n,m}$ with $\pi_{n'}(c_{n', m'}) > 0$ satisfies $\eta_{n'}(c_{n', m'}) \in T$. The probability of termination after $n$ steps is
    \[
        \Pr[B \text{ terminates after } n \text{ steps}] \coloneqq \sum_{c \in C_n} \mathbb{I}[B \text{ terminates after } c] \cdot \pi_n(c),
    \]
    and the expected utility is
    \[
        \mathcal{V}(B) \coloneqq \sup_{n \in \mathbb{N}} \sum_{c \in C_n} \mathbb{I}[B \text{ terminates after } c] \cdot \pi_n(c) \cdot v(\beta_n(c)).
    \]
\end{define}
As stated by the following lemma, every infinite sequential persuasion can be represented by an $h$-branching sequential persuasion. Under Assumption \ref{assume-closed-F}, $h$-branching sequential persuasions actually represent a larger set than the set of infinite sequential persuasions, because it allows some randomness over the choice of experiments: For two children of $c_{n,m}$ denoted by $c, c' \in \Gamma(n, m)$, we do not require $\beta_{n+1}(c) \ne \beta_{n+1}(c')$. Therefore, $c$ and $c'$ can be induced by the same outcome of the experiment $\eta_n(c_{n,m})$ and therefore encode the same history $\xi \in \Xi_{n+1}$, but the $h$-branching sequential persuasion can take different experiments after $c$ and $c'$. However, the randomness of the experiment choice is limited, since only at most $h$ random choices are available after $c_{n,m}$.
\begin{lemma}\label{lem-branch-represent}
    Under Assumption \ref{assume-closed-F}, every $S = (\Xi_0, \phi_1, \Xi_1, \ldots) \in \mathcal{S}_\mu$ can be represented by some $h$-branching sequential persuasion $B = (C_n, \pi_n, \beta_n, \eta_n)_{n \in \mathbb{N}}$ starting from $\mu$, so that for every $n \in \mathbb{N}$, there exists mapping $r_n: C_n \rightarrow \Xi_n$ such that:
    \begin{itemize}
        \item $\Pr[\xi] = \sum_{c \in C_n} \mathbb{I}[r_n(c) = \xi] \cdot \pi_n(c)$ for every $\xi \in \Xi_n$;
        \item $\text{last}(\xi) = \beta_n(c)$ and $\phi_{n+1}(\xi) = \eta_n(c)$ for every $\xi \in \Xi_n$ and every $c \in C_n$ with $r_n(c) = \xi$.
        \item $\xi \in \Xi_n$ is a prefix of $\xi' \in \Xi_{n'}$ if any only if for every $c' \in C_{n'}$ with $r_{n'}(c') = \xi'$, there exists $c \in C_n$ such that $r_n(c) = \xi$ and $c$ is an ancestor of $c'$.
    \end{itemize}
\end{lemma}

To prove the existence of an optimal Markov sequential persuasion, we first prove that there exists an $h$-branching sequential persuasion $B^*$ starting from $\mu$ such that $\mathcal{V}(B^*) = v_\infty(\mu)$. Then, we use Lemma \ref{lem-Markov-optimal} to reduce the existence of an optimal Markov sequential persuasions to the existence of optimal memoryless and deterministic strategies of an MDP with a reachability objective. For the reachability objective, the existence of an optimal strategy of the MDP, which is $B^*$, guarantees the existence of an optimal memoryless and deterministic strategy.

Intuitively, $B^*$ is constructed by step-by-step convergence of some subsequence of a sequence $(S_k^{(n_k)})_{k \ge 1}$ that satisfies Assumption \ref{assume-uniform-as}. For every $n \in \mathbb{N}$, there exists some subsequence whose first $n$ steps converges, with certain proper notion of convergence, to the first $n$ steps of an $h$-branching sequential persuasion. Then for $n=1,2,\ldots$, we can iteratively take subsequence, so that for every $n$, the subsequence whose first $n+1$ steps converge is contained in the subsequence whose first $n$ steps converge. Therefore, the two limiting prefixes of $h$-branching sequential persuasions are identical for the first $n$ steps. This allows us to define an $h$-branching sequential persuasion which infinite steps, whose every finite prefix is the limit of some subsequence.

\begin{lemma}\label{lem-iterate-converge}
    Suppose $F$ satisfies Assumption \ref{assume-closed-F}, and $(S^{(n_k)}_k)_{k \ge 1}$ is a sequence of finite-step sequential persuasions starting from $\mu$ that satisfies the conditions of Assumption \ref{assume-uniform-as}. Extend each $S^{(n_k)}_k$ to an infinite sequential persuasion $S_k$ by repeating trivial experiments after the first $n_k$ steps, represent each $S_k$ by an $h$-branching sequential persuasion $B_k = (C_n, \pi^{(k)}_n, \beta^{(k)}_n, \eta^{(k)}_n)_{n \in \mathbb{N}}$, and denote the obtained sequence by $(B_k)_{k \ge 1}$. Then, there exists an $h$-branching sequential persuasion $B^* = (C_n, \pi^{*}_n, \beta^{*}_n, \eta^{*}_n)_{n \in \mathbb{N}}$ satisfying: for any $\epsilon > 0$ and the corresponding $n_\epsilon$ given in Assumption \ref{assume-uniform-as},
    \[
        \Pr[B^* \text{ terminates after } n_\epsilon \text{ steps}] \ge 1-\epsilon,
    \]
    and,
    for every $n \ge 1$, there exists a subsequence $(B_{k_l})_{l \ge 1}$ such that: for any $1 \le i \le n$, $(\pi^{(k_l)}_n)_{l \ge 1}$ pointwise converges to $\pi^{*}_n$, $(\beta^{(k_l)}_n)_{l \ge 1}$ pointwise converges in total variation to $\beta^{*}_n$, and $(\eta^{(k_l)}_n)_{l \ge 1}$ pointwise weakly converges to $\eta^{*}_n$.
\end{lemma}

When Lemma \ref{lem-iterate-converge}, we are ready to show that under Assumption \ref{assume-closed-F}, an optimal infinite sequential persuasion exists if an only if the sequence in Assumption \ref{assume-uniform-as} exists. Note that the necessity of Assumption \ref{assume-uniform-as} can be trivially obtained by considering the sequence of finite-step sequential persuasions given by the first $n$ steps of the optimal infinite sequential persuasion, for $n=1,2,\ldots$ and using Lemma \ref{lem-AS-terminate}. Moreover, to achieve optimal expected utility, it is without loss of generality to consider Markov sequential persuasions.
\begin{theorem}\label{thm-opt-exist}
    Under Assumption \ref{assume-closed-F}, there exists an optimal infinite sequential persuasion starting from $\mu$ if and only if Assumption \ref{assume-uniform-as} holds. If some optimal infinite sequential persuasion starting from $\mu$ exists, there exists an optimal Markov sequential persuasion starting from $\mu$. 
\end{theorem}
Lemma \ref{lem-iterate-converge} can also be used to prove the following result about the convergence rate of $(v_n)_{n \in \mathbb{N}}$. For $B^*$ starting from $\mu$, we have $\mathcal{V}(B^*) = v_\infty(\mu)$, and the termination rate of $B^*$ satisfies the same $(\epsilon, n_\epsilon)$ as the sequence in Assumption \ref{assume-uniform-as}. By taking the first $n_\epsilon$ steps of $B^*$ and eliminating possible randomness by choosing the candidate with the highest utility, we obtain the following bound. Tight example for any $(\epsilon, n_\epsilon)$ can be easily constructed using $v$ taking value in $\{\underline V, \overline V\}$, and $F$ that only allows the belief to enter areas where $v$ takes value $\overline V$ after wandering outside the areas for $n_\epsilon$ steps.
\begin{coro}
    For $(\mu, F, v)$ satisfying Assumption \ref{assume-closed-F} and \ref{assume-uniform-as}, for any $\epsilon > 0$ and the corresponding $n_\epsilon$ given in Assumption \ref{assume-uniform-as}, we have
    \[
        v_{n_\epsilon}(\mu) \ge v_\infty(\mu) - \epsilon \cdot (\overline V - \underline V).
    \]
\end{coro}

\bibliographystyle{plainnat}
\bibliography{references}

\appendix
\section{Proof of Lemma \ref{lem-AS-terminate}}
\begin{proof}
    Suppose $\lim_{n \to \infty} \Pr[S \text{ terminates after } n \text{ steps}] = 1-\delta$ for some $\delta > 0$. Take a sufficiently small $\epsilon$ such that $\epsilon \cdot \overline V < (\delta + \epsilon) \cdot \underline V$, since the probability of termination after $n$ steps is increasing in $n$, there exists $n_\epsilon > 0$ such that 
    \[
        \lim_{n \to \infty} \Pr[S \text{ terminates after } n \text{ steps}] < \Pr[S \text{ terminates after } n_\epsilon \text{ steps}] + \epsilon.
    \] 
    Consider the $n_\epsilon$-step sequential persuasion $S^{(n_\epsilon)}$ given by the first $n_\epsilon$ steps of $S$, we have
    \begin{align*}
        \mathcal{V}(S) \le& \sum_{\xi \in \Xi_{n_\epsilon}} [\mathbb{I}[S \text{ terminates after } \xi] \cdot \Pr[\xi] \cdot v(\text{last}(\xi))] + \epsilon \cdot \overline V \\
        \le& \sum_{\xi \in \Xi_{n_\epsilon}} [\mathbb{I}[S \text{ terminates after } \xi] \cdot \Pr[\xi] \cdot v(\text{last}(\xi))] + (\delta + \epsilon) \cdot \underline V \\
        =& \mathcal{V}(S^{(n_\epsilon)}).
    \end{align*}
\end{proof}

\section{Proof of Lemma \ref{lem-iterate}}
\begin{proof}
    For $n \ge 1$ and any $S^{(n)} \in \mathcal{S}^{(n)}_p$, consider the experiment $e = (\lambda_j, p_j)_{j \in [m]}$ that $S^{(n)}$ takes at the first step, which is in $\mathcal{F}(p)$. For every $j \in [m]$, the second to the $n$-th steps of $S^{(n)}$ induces an element of $\mathcal{S}^{(n-1)}_{p_j}$ with probability $\lambda_j$. Therefore,
    \[
        \mathcal{V}(S^{(n)}) \le \sum_{j=1}^m \lambda_j v_{n-1}(p_j),
    \]
    and by taking supremum over $\mathcal{S}^{(n)}_p$, the RHS becomes the supremum over all valid first steps starting from $p$, which is $\mathcal{F}(p)$, and this gives 
    \[
        v_{n}(p) \le \sup_{(\lambda_j, p_j)_{j \in [m]} \in \mathcal{F}(p)} \sum_{j=1}^m \lambda_j v_{n-1}(p_j), \ \forall n \ge 1, \forall p \in \Delta(\Omega).
    \]
    Then suppose the LHS is strictly smaller than the RHS. Denote $\epsilon \coloneqq \sup_{(\lambda_j, p_j)_{j \in [m]} \in \mathcal{F}(p)} \sum_{j=1}^m \lambda_j v_{n-1}(p_j) - v_{n}(p) > 0$, there exists $e^* = (\lambda^*_j, p^*_j)_{j \in [m^*]} \in \mathcal{F}(p)$ such that
    \[
        v_{n}(p) < \sum_{j=1}^{m^*} \lambda^*_j v_{n-1}(p^*_j) - \epsilon,
    \]
    and by definition of $v_{n-1}$, for every $j \in [m^*]$, there exists $S^{(n-1)}_j \in S^{(n-1)}_{p^*_j}$ such that $\mathcal{V}(S^{(n-1)}_j) > v_{n-1}(p^*_j) - \epsilon$. Therefore, 
    \[
        \sum_{j=1}^{m^*} \lambda^*_j \mathcal{V}(S^{(n-1)}_j) > \sum_{j=1}^{m^*} \lambda^*_j v_{n-1}(p^*_j) - \epsilon > v_n(p).
    \]
    Consider the $n$-step sequential persuasion given by: Take $e^*$ for the first step, and choose the rest $n-1$ steps according to $S^{(n-1)}_j$ if the belief after $e^*$ is $p_j$, for every $j \in [m^*]$. The expected utility of this $n$-step sequential persuasion is equal to the LHS of the above inequality, which gives
    \[
        v_{n}(p) \ge \sum_{j=1}^{m^*} \lambda^*_j \mathcal{V}(S^{(n-1)}_j),
    \]
    a contradiction. We conclude that 
    \[
        v_{n}(p) = \sup_{(\lambda_j, p_j)_{j \in [m]} \in \mathcal{F}(p)} \sum_{j=1}^m \lambda_j v_{n-1}(p_j).
    \]
\end{proof}

\section{Proof of Lemma \ref{lem-upper-semicont}}
\begin{proof}
    Since $v$ takes values in $[\underline V, \overline V]$, every function in $(v_n)_{n \ge 1}$ and the limit $v_\infty$ also take values in $[\underline V, \overline V]$.
    
    By definition $v_0 = v$ is upper semicontinuous. Assume that for some $k \ge 1$, $v_{k-1}$ is upper semicontinuous, but $v_k$ is not. Then there exists some $p \in \Delta(\Omega)$, some sequence $(p_i)_{i \ge 1}$ that converges in total variation (recall that we assign total variation metric to $\Delta(\Omega)$) to $p$, and some $C \in \mathbb{R}$ s.t. $v_k(p) < C$, $v_{k}(p_i) > C, \forall i \ge n$, and we can let $v_k(p_i) > p_{k-1}(p_i), \forall i \ge 1$ by neglecting all $p_i$ in the sequence with $v_k(p_i) = p_{k-1}(p_i)$. For any $i \ge 1$, since $v_{k}(p_i) > C$, there exists some $e_i = (\lambda_j^{(i)}, p_j^{(i)})_{j \in [m_i]} \in \mathcal{F}(p_i)$ s.t. 
    \[
        \sum_{j \in [m_i]} \lambda_j^{(i)} v_{k-1}(p_j^{(i)}) > C.
    \]
    Prokhorov's theorem guarantees the existence of some subsequence of $(e_i)_{i \ge 1}$ that weakly converges to some $e^* \in \Delta(\Delta(\Omega))$. By Assumption \ref{assume-closed-F}, we have $e^* \in F$ and the finite-support expression $e^* = (\lambda_j^*, p_j^*)_{j \le m^*}$ with $m^* \le h$.
    
    Since $v_{k-1}$ is bounded and upper semicontinuous, and $(e_i)_{i \ge 1}$ weakly converges to $e^*$, we have
    \begin{equation}\label{eq-limit-experiment}
        \sum_{j \in [m^*]} \lambda^*_j v_{k-1}(p^*_j) \ge C.
    \end{equation}
    Moreover, since the projection mapping $I_\omega(p) \coloneqq p(\omega)$ is continuous and bounded, and $\sigma(e_i) = p_i, \forall i$, we have $\sigma(e^*) = p$. This together with equation (\ref{eq-limit-experiment}) gives $ v_k(p) \ge C$, which contradicts to our assumption. We conclude that $v_k$ is upper semicontinuous. By induction over $k$ we have $v_n$ is upper semicontinuous for all $n \in \mathbb{N}$.
    
    The sequence of upper semicontinuous functions $(v_n)_{n \in \mathbb{N}}$ pointwise converges to $v_\infty$. If the convergence is uniform, then $v_\infty$ is also upper semicontinuous. For any prior $\mu \in \Delta(\Omega)$, we have $H(\mu) \le \log(|\Omega|)$. Taking any experiment in $F$ reduces the expected entropy of belief by at least $\delta$. Therefore, for any (finite-step or infinite) sequential persuasion starting from $\mu$ and any $n \ge 1$, the probability that it takes more that $n$ nontrivial experiments is upper bounded by $\log(|\Omega|) / (n\delta)$. For $m > n$, to obtain the supremum in the definition of $v_m$, it is without loss of generality to consider those $m$-step sequential persuasions that only take trivial experiments after termination. For any such infinite sequential persuasion $S^{(m)} \in \mathcal{S}^{(m)}_\mu$, with probability $1-\frac{\log(|\Omega|)}{n\delta}$ only trivial experiments are taken after obtaining $n$-step belief $b_n$, therefore the final belief $b_m$ satisfies:
    \[
        \Pr[b_m \ne b_{n}] \le \frac{\log(|\Omega|)}{n\delta}.
    \]
    Note that by simulating the first $n$ steps of $S^{(m)}$, we obtain a valid $n$-step sequential persuasion $\tilde S^{(n)}$ with
    \[
        \mathcal{V}(S^{(m)}) - \mathcal{V}(\tilde S^{(n)}) \le \frac{\log(|\Omega|)}{n\delta} \cdot \overline V.
    \]
    And taking supremum over $S^{(m)} \in \mathcal{S}^{(m)}_\mu$, then over $\mu \in \Delta(\Omega)$ gives
    \[
        \sup_{\mu \in \Delta(\Omega)} |v_m(\mu) - v_n(\mu)| \le \frac{\log(|\Omega|)}{n\delta} \cdot \overline V.
    \]
    By Cauchy's criterion, $(v_n)_{n \in \mathbb{N}}$ uniformly converges to $v_\infty$.
\end{proof}

\section{Proof of Lemma \ref{lem-Markov-optimal}}
\begin{proof}
    By Definition \ref{define-markov}, $S = (\Xi_0, \phi_1, \Xi_1, \ldots)$ terminates after every history $\xi$ with $\text{last}(\xi) \in Z$. Moreover, by the condition on $\rho$ in the definition, for every $n \in \mathbb{N}$ and every $\xi \in \Xi_n$ with $\Pr[\xi] > 0$, every belief along $\xi$ is in $D \cup Z$, therefore every experiment in $\xi$ is guided by $\rho$, and is therefore exact. We use induction over $n$ to prove the following equality: For every $n \in \mathbb{N}$, the $n$-step belief $b_n$ induced by $S$ satisfies
    \[
        \mathbb{E}[\mathbb{I}[b_n \in Z] \cdot v(b_n)] + \mathbb{E}[\mathbb{I}[b_n \notin Z] \cdot v_\infty(b_n)] = v_\infty(\mu).
    \]
    For $n=0$, $b_n$ is equal to $\mu$ with probability 1. If $\mu \in Z$, by $Z \subseteq N$ the first term on the left hand side is equal to $v(\mu) = v_\infty(\mu)$, and the second term is equal to zero. If $\mu \notin Z$, the first term is equal to zero and the second term is equal to $v_\infty(\mu)$, therefore the equality holds.

    Suppose the equality holds for $n-1$. The belief $b_n$ is given by applying one experiment according to $b_{n-1}$:
    \begin{align*}
        & \mathbb{E}[\mathbb{I}[b_n \in Z] \cdot v(b_n)] + \mathbb{E}[\mathbb{I}[b_n \notin Z] \cdot v_\infty(b_n)] \\
        =& \sum_{\xi \in \Xi_{n-1}} \mathbb{E}_{p \sim \phi_n(\xi)}[\mathbb{I}[p \in Z] \cdot v(p) + \mathbb{I}[p \notin Z] \cdot v_\infty(p)] \cdot \Pr[\xi] \\
        =& \mathbb{E}_{b_{n-1}} \big[ \mathbb{E}_{p \sim \rho(b_{n-1})}[\mathbb{I}[p \in Z] \cdot v(p) + \mathbb{I}[p \notin Z] \cdot v_\infty(p)] \big] \\
        =& \mathbb{E}_{b_{n-1}} \big[ \mathbb{E}_{p \sim \rho(b_{n-1})}v_\infty(p) \big] \\
        =& \mathbb{E} [v_\infty(b_{n-1})] \\
        =& \mathbb{E}[\mathbb{I}[b_{n-1} \in Z] \cdot v(b_{n-1})] + \mathbb{E}[\mathbb{I}[b_{n-1} \notin Z] \cdot v_\infty(b_{n-1})] \\
        =& v_\infty(\mu),
    \end{align*}
    which proves the induction hypothesis for $n$. The third line comes from $\Pr[\xi] > 0$ only if $\text{last}(\xi) \in D \cup Z$, the fourth and the sixth line comes from $Z \subseteq N$, and the fifth line comes from $\rho$ only takes exact experiments. Therefore the equality holds for every $n \in \mathbb{N}$.
    
    Since $S \in \mathcal{S}_\mu$ almost surely terminates, we have $\lim_{n \to \infty} \Pr[b_n \notin Z] = 0$. This together with boundedness of $v_\infty$ gives
    \[
        \mathcal{V}(S) \ge \lim_{n \to \infty} \mathbb{E}[\mathbb{I}[b_n \in Z] \cdot v(b_n)] = v_\infty(\mu),
    \]
    therefore $S$ is optimal.

\end{proof}

\section{Proof of Lemma \ref{lem-branch-represent}}
\begin{proof}
    $r_n$ and $(\pi_n, \beta_n, \eta_n)$ can be iteratively constructed for $n=0,1,2,\ldots$. For $n=0$ it suffices to let $r_0(c_{0,1}) = (\mu)$, $\beta_0(c_{0,1}) = \mu$, $\eta_0(c_{0,1}) = \phi_1((\mu))$, $\pi_0(c_{0,1}) = 1$. 
    
    Suppose we have obtained $r_n$ and $(\pi_n, \beta_n, \eta_n)$. For every $c_{n,m} \in C_n$, $\eta_{n}(c_{n,m}) = (\lambda_j, p_j)_{j \in [m]}$ generates beliefs $(p_j)_{j \in [m]}$ with $m \le h$, and for every $j \in [m]$, we assign each $p_j$ to a unique $c' \in \Gamma(n,m)$ so that $\beta_{n+1}(c') = p_j$, and let $\pi_{n+1}(c') = \pi_n(c_{n,m}) \cdot \lambda_j$. If $m < h$, we can let $\beta_{n+1}$ take arbitrary beliefs in $(p_j)_{j \in [m]}$ for the rest elements of $\Gamma(n,m)$, and let $\pi_{n+1}$ take zero on these elements. This finishes the construction of $\beta_{n+1}$ and $\pi_{n+1}$. $r_{n+1}$ is given by $r_{n+1}(c) \coloneqq r_n(c_{n,m}) \oplus (\eta_n(c_{n,m}), \beta_{n+1}(c))$ for every $c_{n,m} \in C_n$ and every $c \in \Gamma(n,m)$. $\eta_{n+1}$ is given by $\eta_{n+1}(c) \coloneqq \phi_{n+2}(r_{n+1}(c))$ for every $c \in C_{n+1}$.
\end{proof}

\section{Proof of Lemma \ref{lem-iterate-converge}}

\begin{lemma}\label{lem-step-converge}
    Under Assumption \ref{assume-closed-F}, for any $n \in \mathbb{N}$, if we have a sequence of $h$-branching sequential persuasions $(B_{k})_{k \ge 1}$ starting from $\mu$, with each $B_{k} = (C_n, \pi^{(k)}_n, \beta^{(k)}_n, \eta^{(k)}_n)_{n \in \mathbb{N}}$, such that
    \begin{enumerate}
        \item for every $0 \le i \le n-1$, $(\eta^{(k)}_{i-1})_{k \ge 1}$ pointwise weakly converges to some $\eta_{i-1}: C_{i-1} \rightarrow F$;
        \item for every $0 \le i \le n$, $(\pi^{(k)}_i)_{k \ge 1}$ pointwise converges to some probability distribution $\pi_i$ over $C_i$; 
        \item  for every $0 \le i \le n$, $(\beta^{(k)}_i)_{k \ge 1}$ pointwise converges in total variation distance to some $\beta_i: C_i \rightarrow \Delta(\Omega)$.
    \end{enumerate}
    Then, there exists a subsequence indexed by $(k_l)_{l \ge 1}$, a mapping $\eta_n: [h^{n}] \rightarrow F$, a probability distribution $\pi_{n+1}$ over $[h^{n+1}]$, and a mapping $\beta_{n+1}: [h^{n+1}] \rightarrow \Delta(\Omega)$ s.t. 
    \begin{enumerate}
        \item $(\eta^{(k_l)}_{n})_{l \ge 1}$ pointwise weakly converges to $\eta_{n}$.
        \item $(\pi^{(k_l)}_{n+1})_{l \ge 1}$ pointwise converges to $\pi_{n+1}$; 
        \item $(\beta^{(k_l)}_{n+1})_{l \ge 1}$ pointwise converges in total variation to $\beta_{n+1}$;
        \item $\sigma(\eta_n(c_{n, m})) = \beta_n(c_{n, m})$ and $\tau(\eta_n(c_{n, m})) \subseteq \{\beta_{n+1}(c) | c \in \Gamma(n, m)\}$ for every $n \in \mathbb{N}$ and $1 \le m \le h^n$.
        \item For every $n \in \mathbb{N}$ and $1 \le m \le h^n$, for every $p$ assigned with probability $\lambda \ge 0$ by $\eta_n(c_{n, m})$, $\sum_{c \in \Gamma(n, m)} \mathbb{I}[\beta_{n+1}(c) = p] \cdot \pi_{n+1}(c) = \pi_n(c_{n, m}) \cdot \lambda$.
    \end{enumerate}
\end{lemma}
\begin{proof}
    Note that for finite $\Omega$, every infinite sequence in the compact metric space $\Delta(\Omega)$ contains some subsequence that converges in total variation distance. Therefore, we can start from the initial sequence indexed by $(k_l)_{l \ge 1} = 1,2,3,\ldots$, for every $1 \le m \le h^{n+1}$, iteratively take subsequence of the current sequence, where each iteration step makes the obtained subsequence of $\beta_{n+1}^{(k_l)}(c_{n+1,m})$ weakly converge. After $h^{n+1}$ such steps, we obtain a subsequence indexed by $(k^{(1)}_l)_{l \ge 1}$ and function $\beta_{n+1}$ given by the point-wise weak limits of $(\beta_{n+1}^{(k^{(1)}_l)})_{l \ge 1}$ on $[h^{n+1}]$. 
    
    Similarly, since the $(h^{n+1}-1)$-dimensional simplex is compact, we can start from indexes $(k^{(1)}_l)_{l \ge 1}$, iteratively take subsequence to let the subsequence of $\pi_{n+1}(c_{n+1,m})$ converge and assign the limit to $\pi_{n+1}(c_{n+1,m})$, for every $m \in [h^{n+1}]$. After all $h^{n+1}$ iterations, the obtained subsequence is indexed by $(k^{(2)}_l)_{l \ge 1}$.
    
    Now we obtained a sequence $(B_{k^{(2)}_l})_{l \ge 1}$ in which $(\pi_i^{(k^{(2)}_l)})_{l \ge 1}$ pointwise converges to $\pi_i$ and $(\beta_i^{(k^{(2)}_l)})_{l \ge 1}$ pointwise converges to $\beta_i$, for every $i=n, n+1$. Moreover, By definition of $h$-branching sequential persuasion, every $B_{k^{(2)}_l}$ satisfies compatibility of belief transitions and consistency of probability distributions about the $n$-th step transition $\eta_n$. Therefore, we obtain that $(\eta_{n+1}^{(k^{(2)}_l)})_{l \ge 1}$ pointwise weakly converges to some $\eta_n$, and the limiting functions $(\pi_n, \pi_{n+1}, \beta_n, \beta_{n+1}, \eta_n)$ also satisfies compatibility of belief transitions and consistency of probability distributions.
\end{proof}

\begin{proof}[Proof of Lemma \ref{lem-iterate-converge}]
    We construct the limiting $h$-branching sequential persuasion $B^* = (C_n, \pi^{*}_n, \beta^{*}_n, \eta^{*}_n)_{n \in \mathbb{N}}$ step-by-step, by iteratively applying Lemma \ref{lem-step-converge} to take subsequence of $(B_k)_{k \ge 1}$. 

    For $n=0$, $C_0$ contains $c_{0,1}$ as the only element and for every $k \ge 1$, $B_{k} = (C_n, \pi^{(k)}_n, \beta^{(k)}_n, \eta^{(k)}_n)_{n \in \mathbb{N}}$ satisfies $C_n = \{c_{0,1}\}$ and $\beta_0^{(k)}(c_{0,1}) = \mu$, therefore $(\pi_0^{(k)}, \beta_0^{(k)})$ are identical for every $k$. Clearly, the condition of Lemma \ref{lem-step-converge} holds for $n=0$.
    
    For any $n_0 \in \mathbb{N}$, suppose there exists a subsequence of $(B_k)_{k \ge 1}$ given by $(B_{k_l})_{l \ge 1}$ that satisfies the conditions of Lemma \ref{lem-step-converge} for $n=n_0$. By Lemma \ref{lem-step-converge}, we can further extract a subsequence from $(B_{k_l})_{l \ge 1}$ indexed by $(k'_l)_{l \ge 1}$, so that $(\eta^{(k_l)}_{n})_{l \ge 1}$ pointwise weakly converges to some $\eta_{n}: [h^n] \rightarrow F$, $(\pi^{(k_l)}_{n+1})_{l \ge 1}$ pointwise converges to some probability distribution $\pi_{n+1}$ over $[h^{n+1}]$, and $(\beta^{(k_l)}_{n+1})_{l \ge 1}$ pointwise converges in total variation to some $\beta_{n+1}: [h^{n+1}] \rightarrow \Delta(\Omega)$. This satisfies the conditions of Lemma \ref{lem-step-converge}. 
    

    Therefore, for $n=0, 1, 2, \ldots$, we can iteratively apply Lemma \ref{lem-step-converge} to iteratively take subsequence of $(B_k)_{k \ge 1}$, and obtain functions $(\pi_n, \beta_n, \eta_n)_{n \in \mathbb{N}}$. We claim that $B^* = (C_n, \pi_n, \beta_n, \eta_n)_{n \in \mathbb{N}}$ is an $h$-branching sequential persuasion satisfying the requirements of Definition \ref{def-h-branch}, that the belief transitions are compatible and the probability distributions are consistent. This is guaranteed by Lemma \ref{lem-step-converge}, for every $n \in \mathbb{N}$.

    
    The last step is to show $\lim_{n \to \infty} \Pr[B^* \text{ terminates after } n \text{ steps}] = 1$. By Assumption \ref{assume-uniform-as}, given $\epsilon > 0$, there exists $n_\epsilon > 0$ s.t. every infinite sequential persuasion in $(S_k)_{k \ge 1}$ terminates with probability at least $1-\epsilon$ after $n_\epsilon$ steps. We claim that $B^*$ also terminates after $n_\epsilon$ steps with probability as least $1-\epsilon$. To see this, suppose the claim to be false, that is
    \[
        \sum_{c \in C_{n_\epsilon}} \mathbb{I}[B^* \text{ terminates after } c] \cdot \pi_{n_\epsilon}(c) < 1-\epsilon,
    \]
    then there exists $M > 0$ such that
    \[
        \sum_{c \in C_{n_\epsilon}} \mathbb{I}[\exists n_\epsilon \le n \le M, \exists c' \in C_n: c' \text{ is a descendant of }c, \pi_n(c') > 0, \eta_n(c') \notin T] \cdot \pi_{n_\epsilon}(c) < 1-\epsilon,
    \]
    By the result above, we can find a subsequence of $(B_k)_{k \ge 1}$ indexed by $(k_l)_{l \ge 1}$ such that $\pi^{(k_l)}_n$ pointwise converges to $\pi_n$, and $\eta^{(k_l)}_n$ pointwise weakly converges to $\eta_n$, for every $0 \le i \le n$. Therefore, there exists some sufficiently large $L \ge 1$ such that, 
    \[
        \sum_{c \in C_{n_\epsilon}} \mathbb{I}[\exists n_\epsilon \le n \le M, \exists c' \in C_n: c' \text{ is a descendant of }c, \pi^{(k_{L})}_n(c') > 0, \eta^{(k_{L})}_n(c') \notin T] \cdot \pi_{n^{(k_{L})}_\epsilon}(c) < 1-\epsilon,
    \]
    (Note that if the limiting experiment is nontrivial, then some postfix of the weakly convergent sequence of experiments has to be all nontrivial, but it is possible that a trivial experiment is the weak limit of a sequence of nontrivial experiments.) This shows that $\Pr[B_{k_L} \text{ terminates after } n_\epsilon \text{ steps}] < 1-\epsilon$. 
    
    However, $B_{k_L}$ is a representation of the infinite sequential persuasion $S_{k_L}$, and by Lemma \ref{lem-branch-represent}, if $B_{k_L}$ takes nontrivial experiments between step $n_\epsilon$ and step $M$ by some descendant of $c \in C_{n_\epsilon}$ with positive probability, then $S_{k_L}$ also takes nontrivial experiments with positive probability between step $n_\epsilon$ and step $M$ after history $r_n(c)$, therefore $S_{k_L}$ does not terminate after $r_n(c)$. This gives 
    \[
        \Pr[S_{k_L} \text{ terminates after } n_\epsilon \text{ steps}] < 1-\epsilon,
    \]
    which contradicts to Assumption \ref{assume-uniform-as}. We conclude that for any $\epsilon > 0$, $B^*$ terminates after $n_\epsilon$ steps with probability at least $1-\epsilon$, where $n_\epsilon$ is given by Assumption \ref{assume-uniform-as}.
\end{proof}

\section{Proof of Theorem \ref{thm-opt-exist}}
\begin{proof}
    We are going to prove two claims. (i) Under Assumption \ref{assume-closed-F} and \ref{assume-uniform-as}, there exists an optimal Markov sequential persuasion starting from $\mu$. (ii) Under Assumption \ref{assume-closed-F}, given an optimal infinite sequential persuasion, we can construct a sequence of finite-step sequential persuasions satisfying Assumption \ref{assume-uniform-as}.
    
    To prove (i): Suppose Assumption \ref{assume-closed-F} holds, and $(S^{(n_k)}_k)_{k \ge 1}$ satisfies the condition of Assumption \ref{assume-uniform-as}. Then, the $h$-branching sequential persuasion $B^* = (C_n, \pi^{*}_n, \beta^{*}_n, \eta^{*}_n)_{n \in \mathbb{N}}$ introduced in Lemma \ref{lem-iterate-converge} actually satisfies $\mathcal{V}(B^*) = v_\infty(\mu)$. To see this, take an arbitrary $\epsilon > 0$, Assumption \ref{assume-uniform-as} and the property of $B^*$ guarantees the existence of $n_\epsilon \in \mathbb{N}$ such that
    \[
        \Pr[B^* \text{ terminates after } n_\epsilon \text{ steps}] > 1-\epsilon,
    \]
    and
    \[
        \Pr[S_k \text{ terminates after } n_\epsilon \text{ steps}] > 1-\epsilon, \forall k \ge 1.
    \]
    By Lemma \ref{lem-iterate-converge}, there exists subsequence $(B_{k_l})_{l \ge 1}$ which are representations of $(S_{k_l})_{l \ge 1}$ by Lemma \ref{lem-branch-represent}, so that for every $l \ge 1$, the $n_\epsilon$-step belief $b_{n_\epsilon}^{(k_l)}$ has the same distribution as $\beta^{(k_l)})_{n_\epsilon}(c)$ for $c \sim \pi^{(k_l)}_{n_\epsilon}$, therefore $(b_{n_\epsilon}^{(k_l)})_{l \ge 1}$ converges in distribution to $\beta^*_{n_\epsilon}(c)$ for $c \sim \pi^*_{n_\epsilon}$. Since $v$ is upper-semicontinuous, we have
    \[
        \lim_{l \to \infty} \mathbb{E}[v(b_{n_\epsilon}^{k_l})] = \mathbb{E}_{c \sim \pi^*_{n_\epsilon}}[v(\beta^*_{n_\epsilon}(c))].
    \]
    This further gives
    \begin{align*}
        & \sum_{c \in C_{n_\epsilon}}[\mathbb{I}[B^* \text{ terminates after } c] \cdot \pi^*_{n_\epsilon}(c) \cdot v(\beta^*_{n_\epsilon}(c))] \\
        \ge& \mathbb{E}_{c \sim \pi^*_{n_\epsilon}}[v(\beta^*_{n_\epsilon}(c))] - \epsilon \cdot \overline V \\
        =& \lim_{l \to \infty} \mathbb{E}[v(b_{n_\epsilon}^{(k_l)})] - \epsilon \cdot \overline V \\
        \ge & \lim_{l \to \infty} (\mathcal{V}(S^{(n_{(k_l)})}_{(k_l)}) - \epsilon \cdot \overline V) - \epsilon \cdot \overline V \\
        =& v_\infty(\mu) - 2\epsilon \cdot \overline V,
    \end{align*}
    where the last line comes from the requirement of Assumption \ref{assume-uniform-as}, that the limiting expected utility of $(S^{(n_k)}_k)_{k \ge 1}$ is equal to $v_\infty(\mu)$. Therefore, for every $\epsilon > 0$ and the corresponding $n_\epsilon$,
    \begin{align*}
        \mathcal{V} (B^*) =& \sup_{n \in \mathbb{N}} \sum_{c \in C_{n}}[\mathbb{I}[B^* \text{ terminates after } c] \cdot \pi^*_{n}(c) \cdot v(\beta^*_{n}(c))] \\
        \ge& \sum_{c \in C_{n_\epsilon}}[\mathbb{I}[B^* \text{ terminates after } c] \cdot \pi^*_{n_\epsilon}(c) \cdot v(\beta^*_{n_\epsilon}(c))] \\
        \ge& v_\infty(\mu) - 2\epsilon \cdot \overline V, \forall \epsilon > 0,
    \end{align*}
    from which we conclude that $\mathcal{V}(B^*) = v_\infty(\mu)$.
    
    Now we have obtained an $h$-branching sequential persuasion $B^*$ with $\mathcal{V}(B^*) = v_\infty(\mu)$, and we use it to derive the existence of an optimal Markov sequential persuasion. It is easy to obtain that, all experiments that $B^*$ takes with positive probability are exact experiments. Moreover, for every $n \in \mathbb{N}$, $B^*$ terminates after $n$ steps at belief $p$ with positive probability only if $p \in N$. Consider an MDP with countably infinite states given by: the states are the beliefs that $B^*$ visits with positive probability, the initial state is $\mu$, and for every state $p$, the set of randomized transitions that can be chosen at state $p$ are given by (i) stay at $p$ w.p. 1, and (ii) for every experiment $e = (\lambda_j, p_j)_{j \in [m]}$ that $B^*$ takes with positive probability, a randomized transition that transits to state $p_j$ w.p. $\lambda_j$. The objective function is the undiscounted probability of entering set $N \subseteq \Delta(\Omega)$. $S^*$ is an optimal strategy for this MDP (it is easy to verify that the probability of each $c \in C_{n}$'s occurrence gives a pre-measure over finite histories, so that we can apply Carath\'eodory's extension theorem to define the probability for infinite histories, and conclude $\Pr[B^* \text{enters N}] = 1$). When an optimal randomized strategy of countably infinite MDP with reachability objective exists, there exists an optimal strategy that is memoryless and deterministic \cite{MDP-2017,MDP-2020}, which corresponds to an optimal Markov sequential persuasion by our definition. 

    \noindent
    To prove (ii): Given an optimal infinite sequential persuasion $S = (\Xi_0, \phi_1, \Xi_1, \ldots)$, consider the sequence $(S^{(k)})_{k \ge 1}$ where the $k$-th term $S^{(k)} = (\Xi_0, \phi_1, \Xi_1, \ldots, \phi_k, \Xi_k)$ is given by the first $k$ steps of $S$. By Lemma \ref{lem-AS-terminate}, the termination probability of $S^{(k)}$ goes to 1 as $k$ goes to infinity. It can be easily verified that $(S^{(k)})_{k \ge 1}$ satisfies Assumption \ref{assume-uniform-as}.
\end{proof}


\end{document}